\documentclass[pdflatex]{sn-jnl2_new}
\usepackage{natbib}
\jyear{2022}%

\usepackage{adjustbox}
\usepackage{physics}
\usepackage{dsfont}
\usepackage[english]{babel}
\usepackage{bbold}

\theoremstyle{thmstyleone}%
\newtheorem{theorem}{Theorem}
\newtheorem{corollary}[theorem]{Corollary}
\newtheorem{proposition}[theorem]{Proposition}%

\theoremstyle{thmstyletwo}%
\newtheorem{remark}{Remark}%

\theoremstyle{thmstylethree}%

\usepackage{comment}
\DeclareMathOperator{\sgn}{\text{sgn}}

\DeclareMathOperator{\vbold}{\mathbf{v}}
\DeclareMathOperator{\xbold}{\mathbf{x}}
\DeclareMathOperator{\zbold}{\mathbf{z}}
\DeclareMathOperator{\ybold}{\mathbf{y}}

\DeclareMathOperator{\rbold}{\mathbf{r}}
\DeclareMathOperator{\wbold}{\mathbf{w}}

\DeclareMathOperator{\bbold}{\mathbf{b}}
\DeclareMathOperator{\xibold}{\boldsymbol{\xi}}

\DeclareMathOperator{\gammabold}{\boldsymbol{\gamma}}

\DeclareMathOperator{\thetabold}{\boldsymbol{\theta}}

\DeclareMathOperator{\psibold}{\boldsymbol{\psi}}
\DeclareMathOperator{\varepsilonbold}{\boldsymbol{\varepsilon}}
\DeclareMathOperator{\sbold}{\mathbf{s}}
\DeclareMathOperator{\hbold}{\mathbf{h}}
\DeclareMathOperator{\Kbold}{\mathbf{K}}
\DeclareMathOperator{\Mbold}{\mathbf{M}}
\DeclareMathOperator{\Jbold}{\mathbf{J}}
\DeclareMathOperator{\Bbold}{\mathbf{B}}
\DeclareMathOperator{\Abold}{\mathbf{A}}
\DeclareMathOperator{\Sbold}{\mathbf{S}}
\DeclareMathOperator{\Dbold}{\mathbf{D}}
\DeclareMathOperator{\Ebold}{\mathbf{E}}
\DeclareMathOperator{\Fbold}{\mathbf{F}}
\DeclareMathOperator{\Hbold}{\mathbf{H}}
\DeclareMathOperator{\Gbold}{\mathbf{G}}
\DeclareMathOperator{\Qbold}{\mathbf{Q}}

\DeclareMathOperator{\Nbold}{\mathbf{N}}

\DeclareMathOperator{\Zbold}{\mathbf{Z}}
\DeclareMathOperator{\Ibold}{\mathbf{I}}

\DeclareMathOperator{\Lbold}{\mathbf{L}}

\DeclareMathOperator{\Rbold}{\mathbf{R}}
\DeclareMathOperator{\Sigmabold}{\mathbf{\Sigma}}
\DeclareMathOperator{\Lambdabold}{\mathbf{\Lambda}}
\DeclareMathOperator{\0bold}{\mathbf{0}}

\DeclareMathOperator{\bbR}{\mathds{R}}
\DeclareMathOperator{\bbN}{\mathds{N}}
\DeclareMathOperator{\bbC}{\mathds{C}}
\newcommand{\cov}{\mathrm{Cov}}
\newcommand\scalemath[2]{\scalebox{#1}{\mbox{\ensuremath{\displaystyle #2}}}}

\newcommand{\underbrset}[2]{\underset{#1}{\underbrace{#2}}}

\usepackage{graphicx} 
\usepackage[normalem]{ulem}

\begin{document}

\title{The SPDE approach for spatio-temporal datasets with advection and diffusion}

\author*[1]{\fnm{Lucia} \sur{Clarotto}}\email{lucia.clarotto@agroparistech.fr}

\author[2]{\fnm{Denis} \sur{Allard}}
\equalcont{\small These authors contributed equally to this work.}

\author[3]{\fnm{Thomas} \sur{Romary}}
\equalcont{\small These authors contributed equally to this work.}

\author[3]{\fnm{Nicolas} \sur{Desassis}}
\equalcont{\small These authors contributed equally to this work.}

\affil[1]{\small \orgdiv{UMR MIA Paris-Saclay, AgroParisTech, INRAE}, \orgname{Université Paris-Saclay}, \orgaddress{\city{Palaiseau}, \postcode{91120}, \country{France}}}
\affil[2]{\small \orgdiv{Biostatistiques et Processus Spatiaux (BioSP)}, \orgname{INRAE}, \orgaddress{\city{Avignon}, \postcode{84914}, \country{France}}}
\affil*[3]{\small \orgdiv{Centre for geosciences and
geoengineering}, \orgname{Mines Paris, PSL University}, \orgaddress{\city{Fontainebleau}, \postcode{77300}, \country{France}}}

\abstract{In the task of predicting spatio-temporal fields in environmental science using statistical methods, introducing statistical models inspired by the physics of the underlying phenomena that are numerically efficient is of growing interest. Large space-time datasets call for new numerical methods to efficiently process them. The Stochastic Partial Differential Equation (SPDE) approach has proven to be effective for the estimation and the prediction in a spatial context. We present here the advection-diffusion SPDE with first-order derivative in time which defines a large class of nonseparable spatio-temporal models. A Gaussian Markov random field approximation of the solution to the SPDE is built by discretizing the temporal derivative with a finite difference method (implicit Euler) and by solving the spatial SPDE with a finite element method (continuous Galerkin) at each time step. The ``Streamline Diffusion'' stabilization technique is introduced when the advection term dominates the diffusion. Computationally efficient methods are proposed to estimate the parameters of the SPDE and to predict the spatio-temporal field by kriging, as well as to perform conditional simulations. The approach is applied to a solar radiation dataset. Its advantages and limitations are discussed.}


\maketitle

\clearpage


\section{Introduction}
\label{sec:intro}

Many areas of environmental science seek to predict a space-time variable of interest from observations at scattered points in the space cross time domain of study, e.g., among other possible applications, wind
prediction \citep{lenzi2020,huang22}, precipitation forecasting \citep{sigrist2011}, urban air quality inference \citep{paciorek2009}. Among modern techniques proposing efficient methods for estimation and prediction in a spatio-temporal framework, there is a distinction between two possible ways of constructing and treating spatio-temporal models \citep{wikle2010}: either one follows the traditional geostatistical paradigm, using joint space-time covariance functions (see for example \citet{cressie99}, \citet{gneiting2002}, \citet{stein2005}, as well as the recent reviews \citet{porcu202130}, \citet{chen2021}), or one uses dynamical models, including functional time series of surfaces, see for example \citet{wikle99}, \citet{sigrist2012} and
\citet{martinez2023}.

While the theoretical aspects of spatio-temporal geostatistics are well developed \citep{cressie2011}, their implementation faces difficulties. The geostatistical paradigm is computationally expensive for large spatio-temporal datasets, due to the factorization of dense covariance matrices, whose complexity scales with the cube of the number of observations. This well known problem is often referred to as the ``big $n$ problem''  \citep{banerjee2014}. Separable space-time covariance functions have often been used to take advantage of their computational convenience, even when they are not realistic in describing the processes due to the absence of space-time interaction. In most applications, separable models show poorer predictions than nonseparable models, see references above. Recent studies have focused on constructing nonseparable models, which are physically more realistic, albeit computationally more expensive, see \citet{gneiting2002}, \citet{porcu2006}, \citet{salvana2021} and \citet{bourotte2016flexible}, \citet{allard2022fully} in a multivariate context. Nonseparable space-time covariance models can be constructed from Fourier transforms of permissible spectral densities, mixtures of separable models, and partial differential equations (PDEs) representing physical laws \citep{vergara2022general,lindgren2022}. They can be fully symmetric or asymmetric, stationary or nonstationary, univariate or multivariate, in the Euclidean space or on the sphere. See \citet{porcu202130} and \citet{chen2021} for recent comprehensive reviews.

In this paper, we follow the dynamic approach that makes use of physical laws and study models which are defined through Stochastic Partial Differential Equations (SPDEs), where the stochasticity is obtained by adding a random noise as a forcing term. The SPDE approach relies on the representation of a continuously indexed Gaussian Random Field (GRF) as a discretely indexed random process, i.e. a Gaussian Markov Random Field (GMRF, see \citet{rue2005}). Passing from a GRF to a GMRF, the covariance function and the dense covariance matrix are substituted respectively by a neighborhood structure and a sparse precision matrix. Using GMRFs with sparse precision matrices implies computationally efficient numerical methods, especially for matrix factorization. The link between GRF and GMRFs in the purely spatial case has been pioneered by \citet{lindgren2011}, who proposed to construct a GMRF representation of the spatial Mat\'ern field on a triangulated mesh of the domain through the discretization of a diffusion SPDE with a Finite Element Method (FEM). We refer to \citet{bakka2022tuto} for a simple explanation of the FEM applied to the spatial SPDE and to Section \ref{sec:discretization} for a detailed generalization to the spatio-temporal SPDE. 

In the spatial framework, major mathematical and algorithmic advances in the SPDE approach have been made \citep{fulgstad2015, pereira2019, pereira_desassis_allard_2022}, making it possible to efficiently process very large datasets, even in the presence of nonstationarities and varying local anisotropies. The development of SPDE-based approaches to Gaussian processes has led to several practical solutions, among which we find the R package for approximate Bayesian inference R-INLA \citep{rue2009, lindgren2015} that uses SPDEs to sample from spatial and spatio-temporal models. 

When generalizing to the spatio-temporal framework, a direct space-time formulation of the SPDE approach was first suggested in \citet{lindgren2011}, without any precise detail on estimation and prediction. In \citet{cameletti2011}, the SPDE approach was coupled with an AR(1) model in time, leading to a separable space-time model. Nonseparable spatio-temporal models have been elaborated in \citet{sarkka2013}, \citet{krainski2019} and \citet{lindgren2022diffusion} as a spatio-temporal generalization of the diffusion-Mat\'ern model of \citet{lindgren2011}.  
In the approaches overviewed above, the space-time processes are symmetrical in the sense that the spatio-temporal covariance does not change when the sign of the space and/or time lag changes. However, atmospheric and geophysical processes are often asymmetric due to transport effects, such as air and water flows. \citet{sigrist2015} built non-symmetrical and nonseparable space-time Gaussian models as a solution to an advection-diffusion SPDE with computationally efficient algorithms for statistical estimation using fast Fourier transforms and Kalman filters. \citet{sharrock2022} used a similar method, but in an joint online parameter estimation and optimal sensor placement problem. \citet{liu2020} extended this approach to spatially-varying advection-diffusion and non-zero mean source-sink, leading to a space-time covariance which is nonstationary in space. The applicability of these approaches remains difficult however, especially with scattered data, as it relies on the Fourier transform of the data.
\citet{vergara2022general} defined new spatio-temporal models incorporating the physical processes linked to the studied phenomena (advection, diffusion, etc.), but  the estimation of the parameters and the conditioning to the observed data remained unaddressed. 

In this work, we propose a new and efficient  approach for dealing with spatio-temporal SPDEs that includes both a diffusion and an advection term. In contrast to \citet{sigrist2015} and \citet{liu2020}, we make use of the sparse formulation of the spatio-temporal field which is the approximate solution of the SPDE obtained by a combination of FEM and finite differences (FD). This sparse formulation allows fast algorithms for parameter estimation and spatio-temporal prediction. We also treat the case of an advection-dominated SPDE, by introducing the Streamline Diffusion (SD) stabilization term in the SPDE \citep{hughes79}. To the best of our knowledge, this work is the first statistical FEM/FD implementation of spatio-temporal SPDEs with advection.

The paper is organized as follows: Section \ref{sec:spatiotemp_spde} first presents background material on the spatio-temporal SPDE approach. The spatio-temporal advection-diffusion model developed in this paper is presented, along with its discretization. Moreover, the stabilization of advection-dominated SPDEs is introduced. Section \ref{sec:estimation} explores fast and scalable estimation methods, kriging formula for prediction and conditional simulations. Section \ref{sec:application} presents an application of the proposed spatio-temporal SPDE approach to a solar radiation dataset. Section \ref{sec:discussion} discusses the advantages and the limitations of the approach and opens the way to further works.  


\section{The spatio-temporal advection-diffusion SPDE and its discretization}
\label{sec:spatiotemp_spde}

\subsection{Background}
\label{sec:background_spde}

In the SPDE representation, GRFs on $\mathds{R}^d$ are viewed as solutions to specific stochastic partial differential equations  \citep{whittle54,whittle63}. In particular, Gaussian Whittle-Mat\'ern fields, analyzed in details in \citet{lindgren2011} and reviewed in \citet{lindgren2022}, are solutions to
\begin{align}
(\kappa^2 - \Delta)^{\alpha/2} X(\mathbf{\cdot}) = \tau W(\mathbf{\cdot}),
\label{eq:spatial_spde}
\end{align} 
with $\alpha > d/2$ and $\tau >0$. $\Delta = \sum_{i=1}^d \frac{\partial^2}{\partial s_i^2}$ is the Laplacian operator and $W(\cdot)$ is a standard spatial Gaussian white noise, whose definition is briefly recalled.

A white noise $W(\cdot)$ is as a Generalized Random Field (GeRF) that associates to any function $\phi \in L^2(\bbR^d)$ a random variable $W(\phi) \in \bbR$, that satisfies
$$\mathds{E}[W(\phi)] = 0, \quad \forall \phi \in L^2(\bbR^d)$$ 
and
\begin{equation}
\cov[W(\phi_1), W(\phi_2)]=\int_{\bbR^d} \phi_1(\sbold) \phi_2(\sbold) \dd{\sbold}, \quad \forall \phi_1, \phi_2 \in L^2(\bbR^d).
\label{eq:white_noise}
\end{equation}
If, moreover, for any $m\geq1$ and any linearly independent $\phi_1,\dots,\phi_m \in L^2(\bbR^d)$, the random vector $[W(\phi_1),\dots, W(\phi_m)]^\top$ is a Gaussian vector, then $W$ is called \textit{Gaussian white noise}. 

If $\{Z_i\}_{i\in\bbN}$ is a sequence of independent, standard Gaussian variables, then the function $W$ defined over $L^2(\bbR^d)$ by
$$W(\phi) =\sum_{j\in\bbN} Z_j \int_{\bbR^d} \phi e_j \dd{\sbold}, \quad \forall \phi \in L^2(\bbR^d),$$
where $\{e_j\}_{j\in\bbN}$ denotes an orthonormal basis of $L^2(\bbR^d)$, is a Gaussian white noise on $\bbR^d$. 

In principle, GeRFs have only meaning when applied to test functions in some particular functional space, and not necessarily when evaluated in points of the space, but, for an easier reading, we will allow ourselves to write $W(\sbold)$ and $X(\sbold)$.

The covariance function of the Gaussian Whittle-Mat\'ern field solution to Equation \eqref{eq:spatial_spde} is the well known Mat\'ern covariance function
\begin{equation}
\cov(\hbold)=\sigma^2 C^M_\nu(\kappa \norm{\hbold})= \frac{\sigma^2}{2^{\nu-1}\Gamma(\nu)}\left(\kappa \norm{\hbold}\right)^\nu \mathcal{K}_\nu\left(\kappa \norm{\hbold}\right),
\label{eq:matern_spatial} 
\end{equation}
with smoothness parameter $\nu=\alpha-d/2>0$, scale parameter $\kappa$ and variance $\sigma^2=\tau^{2}(4\pi)^{-d/2}\Gamma(\nu) \Gamma(\nu+d/2)^{-1}\kappa^{-2\nu}$. $\mathcal{K}_\nu$ is the modified $2^\text{nd}$ order Bessel function and $\hbold = \mathbf{s}-\mathbf{s}'$ is the spatial lag between two locations $\sbold$ and $\sbold'$ in $\mathds{R}^d$. In particular, when $\nu=1/2$, we get the exponential covariance function and when $\nu\rightarrow +\infty$, after proper renormalization, \eqref{eq:matern_spatial} tends to the Gaussian covariance function \citep{genton2001}.

In \citet{lindgren2011}, the smoothness parameter $\nu$ considered in the Mat\'ern covariance function corresponds to integer values of $\alpha$. When non-integer values of $\alpha$ are introduced in the modeling, the SPDE is said to be fractional. Recent reviews of results and applications of the fractional SPDE approach are available in \citet{xiong2024,bolin2019, roques2022}, but this case will not be treated further in this work.

When generalizing to spatio-temporal processes $X(t,\sbold)$, we consider the framework proposed in  \citet{vergara2022general} for extending the SPDE approach to a wide class of linear spatio-temporal SPDEs. Let us denote $\xibold \in \bbR^d$ a spatial frequency and $\omega \in \bbR$ a temporal frequency. The space-time white noise with unit variance, denoted $W(t,\sbold)$, is characterized by its spectral measure $d\mu_W(\omega,\xibold) = (2\pi)^{-(d+1)} d\xibold d\omega$.
New spatio-temporal models were obtained from known PDEs describing physical processes, such as diffusion, advection, and oscillations with stochastic forcing terms. In particular, \citet{vergara2022general} provided sufficient conditions to the existence and uniqueness of stationary solutions to 
\begin{equation}
    \left[\frac{\partial^\beta}{\partial t^\beta} + \mathcal{L}_g\right] X(t,\sbold) = W(t,\sbold),
    \label{eq:evol_eq}
\end{equation}
with $\beta > 0$. In \eqref{eq:evol_eq}, the spatial operator $\mathcal{L}_g$ is defined using the spatial Fourier transform on $\bbR^d$, denoted $\mathcal{F}_S$,
$$\mathcal{L}_g(\cdot)=\mathcal{F}_S^{-1}(g \ \mathcal{F}_S(\cdot)),$$
where  $g : \bbR^d \to \bbC$ is a sufficiently regular and Hermitian-symmetric function called the \textit{symbol function} of the operator $\mathcal{L}_g$. The temporal operator $\frac{\partial^\beta}{\partial t^\beta}$ is 
$$\frac{\partial^\beta}{\partial t^\beta}(\cdot)=\mathcal{F}_T^{-1}((i\omega)^\beta\mathcal{F}_T(\cdot)),$$
where $\mathcal{F}_T$ is the temporal Fourier transform on $\mathds{R}$ and where we have used the symbol function over $\bbR$
$$\omega \mapsto (i\omega)^\beta = \lvert \omega\rvert^\beta e^{i\sgn(\omega)\beta\pi/2}.$$
The spatio-temporal symbol function of the operator involved in \eqref{eq:evol_eq} is thus
\begin{equation*}
    (\omega,\xibold) \mapsto (i\omega)^\beta + g(\xibold) = \lvert\omega\rvert^\beta \cos\left(\frac{\beta \pi}{2}\right) + g_R(\xibold) + i\left(\sgn(\omega)\lvert\omega\rvert^\beta \sin\left(\frac{\beta \pi}{2}\right) + g_I(\xibold)\right),
\end{equation*}
where $g_R$ and $g_I$ are the real and imaginary part of the spatial symbol function $g(\xibold)$. If $\lvert g_R\rvert$ is inferiorly bounded by the inverse of a strictly positive polynomial and $g_R\cos\left(\frac{\beta \pi}{2}\right)\geq 0$, Theorem 1 and Proposition 3 in \citet{vergara2022general} state that \eqref{eq:evol_eq} admits a unique stationary solution for every arbitrary $g_I$ function . 

\subsection{The spatio-temporal advection-diffusion SPDE}
\label{sec:model}
The advection-diffusion equation is a Partial Differential Equation (PDE) that describes physical phenomena where particles, energy, or other physical quantities evolve inside a physical system due to two processes: diffusion and advection. Advection represents the mass transport due to the average velocity of all particles, and diffusion represents the mass transport due to the instantaneously varying velocity of individual particles. 
In this paper, we study the advection-diffusion SPDE on the domain $[0,T]\times\bbR^d$ that writes
\begin{equation}
    \left[\frac{\partial}{\partial t} + \frac{1}{c}(\kappa^2  - \nabla \cdot \Hbold\nabla)^{\alpha} + \frac{1}{c}\gammabold \cdot \nabla \right] X(t,\sbold) = \frac{\tau}{\sqrt{c}} Z(t,\sbold),
    \label{eq:adv_diff}    
\end{equation}
where
\begin{itemize}
    \item the operator $\nabla \cdot \Hbold\nabla$ is a \textit{diffusion} term that can incorporate \textit{anisotropy} in the matrix $\Hbold$. When the field is isotropic, i.e. when $\Hbold = \lambda\Ibold$, this term reduces to the Laplacian operator $\lambda\Delta$;
    \item the operator $\gammabold \cdot \nabla$ models the \textit{advection}, $\gammabold \in \bbR^d$ being a velocity vector;
    \item $\alpha \geq 0$ relates to the smoothness of $X(t,\cdot)$, $\kappa^2>0$ accounts for \textit{damping} and $c$ is a positive time-scale parameter;
    \item $\tau \geq 0$ is a standard deviation factor and $Z$ is a stochastic forcing term, detailed below. 
\end{itemize}

This equation was mentioned in \citet{lindgren2011}, \citet{vergara2022general} and \citet{lindgren2022diffusion}, and was analyzed using spectral approaches in \citet{sigrist2015} and \citet{liu2020}. The term $Z(t,\sbold)$ is assumed to be of the form
\begin{equation}
Z(t,\mathbf{s}) = W_T(t) \otimes  Z_S(\mathbf{s}),
\label{eq:noise}
\end{equation}
i.e., a space-time separable stochastic (Generalized) random function given as the tensor product of a temporal Gaussian white noise $W_T$ and a spatial noise $Z_S$. $Z_S$ is often chosen to be a spatial Gaussian white noise, denoted $W_S$ in this case. To ensure a sufficient regularity for $Z$, $Z_S$ can alternatively be a \textit{colored noise}, such as for example the solution to the spatial Whittle-Mat\'ern SPDE \citep{lindgren2011}
\begin{equation}
(\kappa^2 - \nabla \cdot \Hbold\nabla)^{\alpha_S/2} Z_S(\mathbf{s}) = W_S(\mathbf{s}),
\label{eq:col_noise}
\end{equation}
where $W_S$ is a Gaussian white noise. Note that the parameters $\kappa^2$ and $\Hbold$ in the noise term have been set identical to those in the diffusion term in the left-hand-side of \eqref{eq:adv_diff} to ensure that the spatial marginalization of the process is a Mat\'ern field, as detailed below. A relation can be found between the SPDE notation of our paper and the more classical notation of infinite dimensional  SDEs of \citet{daprato1992}. This relation is explained in Appendix \ref{sec:daprato}.

\medskip

When $\alpha>0$, $X(t,\sbold)$ is a stationary nonseparable spatio-temporal field with covariance function $\cov(u,\hbold)$, with $(u,\hbold) \in \mathds{R}\times\mathds{R}^d$. 
The advection-diffusion equation \eqref{eq:adv_diff} is a particular first-order evolution model as in Equation \eqref{eq:evol_eq} with $\beta=1$. Its spatial symbol function 
$$g(\xibold) = \frac{1}{c}\left[(\kappa^2  + \xibold^\top \Hbold \xibold)^\alpha + i\gammabold^\top\xibold\right]$$
verifies the sufficient condition for existence and uniqueness of a stationary solution recalled at the end of Section \ref{sec:background_spde}. We define the spatial trace of $X$ as the spatial random field $X(t,\cdot)$ at any $t \in [0,T]$.
\citet{vergara2022general} showed that the advection term does not affect the spatial trace of the solution. For some specific values of the parameters, the spatial trace of the solution to \eqref{eq:adv_diff} is a Mat\'ern field, as detailed in Proposition \ref{prop:spatial_trace}. In the following $\lvert \Hbold \rvert$ denotes the determinant of the square matrix $\Hbold$.

\begin{proposition}
\label{prop:spatial_trace}
Let $Z(t,\sbold)$ be a spatio-temporal noise colored in space with $Z_S(\sbold)$ satisfying \eqref{eq:col_noise}, and let  $\alpha_{\text{tot}}=\alpha + \alpha_S$. If $\alpha_{\text{tot}}>d/2$, the spatial trace of the stationary solution $X(t,\sbold)$ of the SPDE \eqref{eq:adv_diff} is the Gaussian Mat\'ern field with covariance 
\begin{equation}
	\cov(0,\hbold) =  \frac{\tau^2\Gamma(\alpha_{\text{tot}}-d/2)}{2\Gamma(\alpha_{\text{tot}})(4\pi)^{d/2}\kappa^{2(\alpha_{\text{tot}}-d/2)}\lvert\Hbold
\rvert^{1/2}}
C^M_{\alpha_{\text{tot}}-d/2}\left(\kappa \norm{\Hbold^{-1/2}\hbold}\right),
	\label{eq:spat_trace}
\end{equation}
where $\hbold=\sbold-\sbold'$ is the spatial lag and $C^M_{\alpha_{\text{tot}}-d/2}(\cdot)$ is the unit variance and scale Mat\'ern covariance function defined in \eqref{eq:matern_spatial} with smoothness parameter equal to $\nu = \alpha_{\text{tot}}-d/2$. 
\end{proposition}

Proposition \ref{prop:spatial_trace} is adapted from Proposition 1 in \citet{lindgren2022diffusion}. A proof is reported in Appendix \ref{sec:A0}. The model reduces to a separable one in a particular case stated in the corollary below.

\begin{corollary}
\label{prop:spatial_trace_sep}
Let the coefficients of the SPDE \eqref{eq:adv_diff} be such that $\alpha=0$ and $\gammabold=\0bold$; the spatial operator applied to the spatio-temporal field $X(t,\sbold)$ is then the constant value $c^{-1}$. Let $Z(t,\sbold)$ be a spatio-temporal noise colored in space, with $Z_S(\sbold)$ satisfying \eqref{eq:col_noise}. If $\alpha_S>d/2$, the stationary solution of the SPDE is a separable spatio-temporal field with covariance
$$\cov(u,\hbold) = \frac{\tau^{2}\Gamma(\alpha_S -d/2)}{2\Gamma(\alpha_S)(4\pi)^{d/2}\kappa^ {2(\alpha_S-d/2)}\lvert\Hbold
\rvert^{1/2}}
C^M_{\alpha_{S}-d/2}\left(\kappa \norm{\Hbold^{-1/2}\hbold}\right) \exp{\frac{u}{c}},$$
with smoothness parameter equal to $\nu=\alpha_S-d/2$.
\end{corollary}


\subsection{Discretization}
\label{sec:discretization}

The advection-diffusion SPDE in \eqref{eq:adv_diff} is discretized in time and space, using finite differences and a finite element method, respectively (from now on, this type of discretization will be noted as FEM/FD discretization). The temporal domain $[0,T]$ is discretized in $(N_T+1)$ regular time steps of length $dt = T/N_T$, and we note $t_k=kdt$ for $k\in\{0,\dots,N_T\}$. Since implicit solvers are usually less sensitive to numerical instability than explicit solvers, the implicit Euler scheme is chosen for the temporal discretization. This choice implies stability, hence convergence towards the stationary solution. We denote $X^{(k)}=X(t_k,\cdot)$ the temporal approximation of the spatial trace of the SPDE \eqref{eq:adv_diff} at time $t_k$. The FEM method for the spatial discretization is the continuous Galerkin method with Neumann Boundary Conditions as detailed in \citet{lindgren2011}.

The solution in two dimensions is now detailed. The solution in three dimensions involve geometrical technicalities, but is otherwise very similar. Let $\Omega \subset \mathds{R}^2$ be a compact and connected domain of $\mathds{R}^2$. $\Omega$ is meshed using a triangulation $\cal T$ with $N_S$ vertices $\{\sbold_1,\dots,\sbold_{N_S}\} \subset \Omega$. Let $ h := \max_{{\rm Tr} \in {\cal T}} h_{\rm Tr}$, where
$h_{\rm Tr}$ is the length of the longest side of the triangle ${\rm Tr} \in {\cal T}$. A first-order finite element representation $X_h$ of the solution to the spatial SPDE is a linear combination  $X_h = \sum_{i=1}^{N_S} x_i\psi_i$ of piecewise linear basis functions $\{\psi_i\}_{i=1}^{N_S}$, each $\psi_i$  being equal to $1$ at the vertex $\sbold_i$ and $0$ at all the other vertices. The weights $\{x_i\}_{i=1}^{N_S}$ define uniquely the values of the field at the vertices, while the values in the interior of the triangles are determined by linear interpolation. The continuous Galerkin solution is then obtained by finding the weights that fulfill the weak  formulation of Equation \eqref{eq:adv_diff} for test functions belonging to the space $\mathcal V$ spanned by $\{\psi_i\}_{i=1}^{N_S}$. 

\begin{proposition}
\label{prop:discr_adv_diff}
Let $X(t,\sbold)$ be the spatio-temporal process solution to Equation \eqref{eq:adv_diff} with $\alpha \in \{0,1\}$ and spatio-temporal white noise, i.e. $Z(t,\sbold) = W(t,\sbold) = W_T(t) \otimes  W_S(\mathbf{s})$. Let $\mathcal{T}$ be a triangulation of $\Omega$ and $\{\psi_i\}_{i=1}^{N_S}$ be the piecewise linear basis functions defined over $\mathcal{T}$. Let us define the mass matrix $\Mbold=[M_{ij}]_{i,j=1}^{N_S}$, the stiffness matrix $\Gbold=[G_{ij}]_{i,j=1}^{N_S}$, the advection matrix $\Bbold=[B_{ij}]_{i,j=1}^{N_S}$ and the matrix $\Kbold=[K_{ij}]_{i,j=1}^{N_S}$ as follows:
\begin{align*}
    M_{ij}&= \int_\Omega \psi_i(\sbold) \psi_j(\sbold)\dd{\sbold}, \nonumber \\
    G_{ij}&= \int_\Omega  \Hbold\nabla\psi_i(\sbold)\cdot \nabla\psi_j(\sbold)\,\dd{\sbold}, \\ 
    B_{ij}&= \int_\Omega  \gammabold\cdot\nabla\psi_i(\sbold)\psi_j(\sbold)\dd{\sbold}, \\
    K_{ij} &= (\kappa^2 M_{ij} + G_{ij})^\alpha.
\end{align*}
Then, at each time step, the continuous Galerkin finite element solution vector $\xbold^{(k+1)} = \{x_i^{(k+1)}\}_{i=1}^{N_S}$, for $k\in\{0,\dots,N_T\}$, satisfies
\begin{equation}
\left(\Mbold+\frac{dt}{c}(\Kbold+\Bbold)\right)\mathbf{x}^{(k+1)}=\Mbold\mathbf{x}^{(k)}+\frac{\tau\sqrt{dt}}{\sqrt{c}} \Mbold^{1/2}\mathbf{z}^{(k+1)},
\label{eq:implicit_scheme}
\end{equation}
where $\mathbf{z}^{(k+1)} \sim \mathcal{N}(\0bold,\Ibold_{N_S})$,  $\Mbold^{1/2}$ is any matrix such that $\Mbold^{1/2}\Mbold^{1/2} = \Mbold$ and $dt=T/N_T$. 
When the noise on the right-hand side is colored in space, i.e. $Z(t,\sbold) = W_T(t) \otimes  Z_S(\mathbf{s})$, the discretization reads
\begin{equation*}
	\left(\Mbold+\frac{dt}{c}(\Kbold+\Bbold)\right)\mathbf{x}^{(k+1)}=\Mbold\mathbf{x}^{(k)}+\frac{\tau\sqrt{dt}}{\sqrt{c}} \Mbold \Lbold_S^\top\mathbf{z}^{(k+1)},
\end{equation*}
where $\Lbold_S$ is the Cholesky decomposition of $\Qbold_{S}^{-1}$, the covariance matrix of the discretized solution $\Zbold_S$ of the spatial SPDE \eqref{eq:col_noise}, obtained with the continuous Galerkin FEM \citep{lindgren2011}.
\end{proposition}

\begin{proof}
The proof is available in Appendix \ref{sec:A1}.
\end{proof}

\begin{remark}
The elements of the matrices $\Mbold$, $\Gbold$ and $\Bbold$ are non-zero only for pairs of basis functions which share common triangles. This implies that the matrix $(\Mbold+\frac{dt}{c}(\Kbold+\Bbold))$ is sparse and that Equation \eqref{eq:implicit_scheme} can be solved by Cholesky decomposition in an efficient  way. 
\end{remark}


\subsection{Stabilization of advection-dominated SPDE}
\label{sec:adv_domination}

When the advection term is too strong with respect to the diffusion term, advection-domination occurs. In the framework outlined above, when $\alpha=1$,
the non-symmetric matrix $\left[\Mbold + \frac{dt}{c}(\Kbold +\Bbold)\right]$ becomes ill-conditioned, which induces oscillations and unstable solutions for the continuous Galerkin approximation. Specifically, the advection-domination occurs when the P\'eclet number $\text{Pe}^h = \frac{\norm{\gammabold} h}{2\lambda} > 1$, where $\lambda$ is the coefficient of the isotropic Laplacian operator (see for example \citet{mekuria2016} or \citet[Chapter 5]{quarteroni2008}).

One possible solution is to decrease the mesh size $h$, i.e., to refine the triangulation, until the advection no longer dominates on the element-level, with $\text{Pe}^h < 1$. However, in many cases this is not a feasible solution because it would increase the number of vertices beyond computation limits. Another solution, adopted here, is to introduce a stabilization term. Many stabilization approaches are possible, some  being more accurate than others \citep[Chapter 5]{quarteroni2008}. 
In our case, we opt for the streamline diffusion stabilization approach \citep{hughes79}, considered as a good trade-off between accuracy and computational complexity. Essentially, the SD approach consists in stabilizing the advection by introducing an artificial diffusion term along the advection direction. The following proposition presents the stabilized solution to \eqref{eq:adv_diff}.

\begin{proposition}
\label{prop:discr_adv_diff_stab}
Assume the same hypotheses as in Proposition \ref{prop:discr_adv_diff} with $\alpha=1$. The solution to Equation \eqref{eq:adv_diff} in presence of 
SD stabilization is
\begin{eqnarray}
\left(\Mbold+\frac{dt}{c}(\Kbold+\Bbold+\Sbold)\right) \mathbf{x}^{(k+1)}=\Mbold\mathbf{x}^{(k)} + \frac{\tilde\tau\sqrt{dt}}{\sqrt{c}}\Mbold^{1/2} \mathbf{z}^{(k+1)},
\label{eq:implicit_scheme_stab}
\end{eqnarray} 
where $\Sbold = [S_{ij}]_{i,j=1}^{N_S}$ is the matrix of the SD stabilization operator $\mathcal{S}$, such that $$S_{ij} = \mathcal{S}(\psi_i,\psi_j)=h\norm{\gammabold}^{-1}\int_\Omega (\gammabold\cdot \nabla \psi_i )(\gammabold\cdot \nabla \psi_j) \dd{\sbold},$$ and $\tilde\tau= \tau \left( \lvert\Hbold + h\norm{\gammabold}^{-1}\gammabold\gammabold^\top \rvert \right)^{-1/4} \left( \lvert\Hbold\rvert \right)^{1/4}$.
When the noise on the right-hand side of Equation \eqref{eq:adv_diff} is colored in space, the discretization becomes
\begin{equation*}
\left(\Mbold+\frac{dt}{c}(\Kbold+\Bbold+\Sbold)\right) \mathbf{x}^{(k+1)}=\Mbold\mathbf{x}^{(k)} + \frac{\tilde\tau\sqrt{dt}}{\sqrt{c}} \Mbold \Lbold_S^\top \mathbf{z}^{(k+1)},
\end{equation*}
where $\Lbold_S$ is as in Proposition \ref{prop:discr_adv_diff}. 
\end{proposition}

The proof of the discretized equation follows the same reasoning as that of Proposition \ref{prop:discr_adv_diff} with the addition of the  matrix $\Sbold$. The streamline diffusion approach can be seen as a perturbation of the original SPDE \citep{bank1990}. Indeed, by making the classical hypothesis of Neumann boundary condition on $\Omega$ and by using the Green's first identity, we get
$$\int_\Omega(\gammabold\cdot \nabla x )(\gammabold\cdot \nabla v) \dd{\sbold} = -\int_\Omega\nabla \cdot (\gammabold\gammabold^\top) \nabla x v\dd{\sbold}.$$
As a consequence, the original SPDE \eqref{eq:adv_diff} can be rewritten with an additional diffusion term as 
\begin{equation}
    \left[\frac{\partial}{\partial t} + \frac{1}{c} \left[\kappa^2 - \nabla \cdot \left(\Hbold + h\norm{\gammabold}^{-1}\gammabold\gammabold^\top\right) \nabla + \gammabold \cdot \nabla \right]\right] X(t,\sbold) = \frac{\tau}{\sqrt{c}}Z(t,\sbold).
    \label{eq:diff_adv_sl_spde}
\end{equation}
The term $(h\norm{\gammabold}^{-1}\gammabold\gammabold^\top)$ acts as an anisotropic “diffusion” matrix that is added to the anisotropy (or identity) matrix $\Hbold$ of the original diffusion. This extra diffusion stabilizes the advection directed along the direction $\gammabold$. By following the proof of Proposition \ref{prop:spatial_trace}, we find that the marginal variance of the spatial field $X(t,\cdot)$ of Equation \eqref{eq:diff_adv_sl_spde} is equal to 
$$\sigma^2 = \frac{\tau^{2}\Gamma(\alpha_{\text{tot}}-d/2)}{\Gamma(\alpha_{\text{tot}})2(4\pi)^{d/2} \kappa^{2(\alpha_{\text{tot}}-d/2)} \lvert\Hbold+h\norm{\gammabold}^{-1}\gammabold\gammabold^\top\rvert^{1/2}}.$$
For the variance to be equal to the variance in Proposition \ref{prop:spatial_trace},  $\tau$ must be replaced by $\tilde{\tau}=\tau \left( \lvert\Hbold + h\norm{\gammabold}^{-1}\gammabold\gammabold^\top \rvert \right)^{1/4} \left( \lvert\Hbold\rvert \right)^{-1/4}$.


\subsection{Spatio-temporal Gaussian Markov Random Field approximation}
\label{sec:prec_mat}

\begin{proposition}
\label{prop:grmf}
In presence of an advection-dominated flow and a spatio-temporal white noise on the right-hand side of Equation \eqref{eq:adv_diff}, the discretized vector $\xbold^{(k+1)}$ on the mesh $\cal T$ at each time step is the solution of the following equation:
\begin{align}
\mathbf{x}^{(0)} &\sim \mathcal{N}(\0bold,\Sigmabold), \nonumber\\
\mathbf{x}^{(k+1)}&=\Dbold\mathbf{x}^{(k)} + \Ebold \mathbf{z}^{(k+1)},
\label{eq:timestep}
\end{align}
where 
\begin{eqnarray}
\Dbold & = & \left(\Mbold + \frac{dt}{c}(\Kbold + \Bbold + \Sbold)\right)^{-1}\Mbold, \nonumber \\
\Ebold & = & \frac{\tilde\tau \sqrt{dt}}{ \sqrt{c}}\left(\Mbold + \frac{dt}{c}(\Kbold + \Bbold + \Sbold)\right)^{-1}\Mbold^{1/2},
\label{eq:matrices_timestep}
\end{eqnarray}
and $\mathbf{z}^{(k+1)}\sim\mathcal{N}(\0bold,\Ibold_{N_S})$ is independent of $\mathbf{x}^{(0)},\dots,\mathbf{x}^{(k+1)}$. 
In presence of a spatio-temporal noise colored in space on the right-hand side of Equation \eqref{eq:adv_diff}, the matrix $\Ebold$ reads 
$$\Ebold = \frac{\tilde\tau \sqrt{dt}}{ \sqrt{c}}\left(\Mbold + \frac{dt}{c}(\Kbold + \Bbold + \Sbold)\right)^{-1}\Mbold\Lbold_S^\top,$$
where $\Lbold_S$ is defined in Proposition \ref{prop:discr_adv_diff}.
\end{proposition}

\begin{proof}
Starting from Equation \eqref{eq:implicit_scheme_stab}, which represents the numerical scheme for the advection-diffusion spatio-temporal SPDE with stabilization, it is straightforward to obtain \eqref{eq:timestep}.
\end{proof}

When the SPDE is not advection-dominated, which implies that no stabilization term is needed, Equation \eqref{eq:matrices_timestep} is replaced by the similar equation where the matrix $\Sbold$ is deleted and $\tilde\tau$ is replaced by $\tau$.

The covariance matrix $\Sigmabold$ of the approximate spatial trace $\xbold^{(0)}$ at the first time step, can be taken to be equal to any admissible positive definite matrix. The closer $\Sigmabold$ is to the covariance $C_S$ of $X(t,\cdot)$, the faster the stationary solution is obtained. When the hypotheses of Proposition \ref{prop:spatial_trace} are satisfied, an efficient option is to choose $\Sigmabold$ as the Mat\'ern covariance of Equation \eqref{eq:spat_trace}. 

To obtain fast inference and prediction computations, the precision matrix of the spatio-temporal discretized solution $\xbold_{0:N_T}=[\mathbf{x}^{(0)},\dots,\mathbf{x}^{(N_T)}]^\top$ must be sparse. For this reason, $\Mbold$ is replaced by the diagonal matrix $\widetilde{\Mbold}=[\widetilde{M}_{ij}]_{i,j=1}^{N_S}$, where $\widetilde{M}_{ii} = \langle \psi_i,1 \rangle$ and $\widetilde{M}_{ij} = 0$ if $i\neq j$ \citep{lindgren2011}. This technique is called mass lumping and is common practice in FEM \citep[Chapter 5]{quarteroni2008}. From now on, we always use the diagonal matrix $\widetilde{\Mbold}$, but for ease of reading, it will still be denoted  $\Mbold$. 

\begin{proposition}\label{prop:global_prec}
Let $\xbold_{0:N_T}=[\mathbf{x}^{(0)},\dots,\mathbf{x}^{(N_T)}]^\top$ be the vector containing all spatial solutions until time step $N_T$ of Equation \eqref{eq:timestep}. The global precision matrix $\Qbold$ of the vector $\xbold_{0:N_T}$ of size $(N_S(N_T+1),N_S(N_T+1))$ reads
\begin{equation}
\scalemath{0.8}{
\Qbold = \begin{pmatrix}
\Sigmabold^{-1}+\Dbold^\top\Fbold^{-1}\Dbold \phantom{X} & -\Dbold^\top\Fbold^{-1} & 0 & \dots & 0 \\
-\Fbold^{-1}\Dbold & \Fbold^{-1}+\Dbold^\top\Fbold^{-1}\Dbold \phantom{X} & -\Dbold^\top\Fbold^{-1}  & \ddots & \vdots \\
\vdots & \ddots & \ddots & \ddots & 0 \\
\vdots & \ddots & -\Fbold^{-1}\Dbold & \phantom{X} \Fbold^{-1}+\Dbold^\top\Fbold^{-1}\Dbold \phantom{X} & -\Dbold^\top\Fbold^{-1} \\
0 & \dots  & 0 & -\Fbold^{-1}\Dbold & \Fbold^{-1}
\end{pmatrix},}
\label{eq:prec_matrix}
\end{equation}
where $\Fbold=\Ebold \Ebold^\top$.
\end{proposition}

\begin{proof}
The proof is available in Appendix \ref{sec:A3}.
\end{proof}

\begin{remark}
    Matrix \eqref{eq:prec_matrix} has a tridiagonal structure in time and is sparse in each of its spatial blocks of size $(N_S,N_S)$ located on the three diagonals. The sparsity of the precision matrix relies on the choice of $\widetilde{\Mbold}$ and on the value of $\alpha$ (the higher $\alpha$ the less sparse the spatial blocks). Conversely, the precision of finite element method discretization is influenced by the mesh density (the smaller $h$, the more precise the solution); this factor plays a role in defining the size of the precision matrix, but not its sparsity pattern.
    The sparsity pattern will be useful for the following sections concerning estimation and prediction. 
\end{remark}

\section{Estimation, prediction and simulation}
\label{sec:estimation}

This section presents an efficient implementation for parameter estimation and spatio-temporal prediction within the spatio-temporal SPDE framework described in Section \ref{sec:spatiotemp_spde}. We consider the advection-diffusion SPDE \eqref{eq:adv_diff} with $d=2$, $\alpha=1$, $\Hbold = \Ibold$ (isotropic diffusion) and colored noise in space with $\alpha_S=2$. Similar computations can be generalized to other values of $\alpha_S$ such that $\alpha_S/2$ is integer or to anisotropic diffusion.

The spatio-temporal domain $\Omega \times [0,T]$ is discretized in space with a triangulation $\cal T$ with $N_S$ nodes and discretized in time by means of $(N_T+1)$ regular time steps. This space-time discretization  is denoted ${\cal T}' = {\cal T} \times \{0,T/N_T,\dots,T\}$. At each time step $k\in\{0,\dots,N_T\}$ there are $n^{(k)}$ observations scattered in the spatial domain $\Omega$. There is thus a total of $n=\sum_{k=0}^{N_T} n^{(k)}$ spatio-temporal data collected in the vector $\ybold_{0:N_T}=[(\mathbf{y}^{(0)})^\top,\dots,(\mathbf{y}^{(N_T)})^\top]^\top$. 

We consider a statistical model with fixed and random effects. The fixed effect is a linear trend on a set of covariates and the random effect is modeled as the FEM/FD discretization of the random field solution to the SPDE \eqref{eq:adv_diff} (as described in Section \ref{sec:discretization}), with the addition of random noise:
\begin{equation}
  \ybold_{0:N_T}=\mathbf{\Lambdabold}\bbold + \mathbf{A}\xbold_{0:N_T} + \sigma_0 \varepsilonbold,
  \label{eq:stat_model}
\end{equation}
where $\bbold$ is the vector of $q$ fixed effects and $\Lambdabold$ is a $(n,q)$ matrix of covariates with $[\Lambdabold]_{jk} = \lambda_k(t_j,\sbold_j)$, $j=1\dots,n$ and $k=1,\dots,q$. The matrix $\Abold$ is the $(n,N_S(N_T+1))$ projection matrix between the data and the points in ${\cal T}'$, and $\varepsilonbold$ is a standard Gaussian random vector with independent components. When the observation locations do not change during the time window, $\Abold^\top\Abold$ is a $(N_S(N_T+1),N_S(N_T+1))$ block-diagonal matrix with all $(N_S,N_S)$ equal blocks.

The discretization with FEM in space and FD in time is justified by a few considerations: compared to a Fourier approximation approach (such as the ones considered in \citet{sigrist2015,liu2020,sharrock2022}), it is better suited for scattered data, it can be easily adapted to spatially and/or temporally varying parameters in the SPDE, leading to nonstationary extensions of the method, and it can be generalized to spatio-temporal fields on Riemannian surfaces with only a few changes in the approach \citep{pereira2023hal}.

\subsection{Estimation of the parameters}
\label{sec:estimation_scattered}

The parameters of the SPDE are estimated using maximum likelihood. We collect the parameters of the SPDE in the vector $\thetabold^\top= [\kappa, \gamma_x, \gamma_y, c, \tau]$, while all the parameters of the statistical model are collected in $\psibold^\top = [\thetabold^\top,\bbold^\top,\sigma_0]$.
Following \eqref{eq:stat_model}, $\ybold_{0:N_T}$ is a Gaussian vector with expectation $\Lambdabold\bbold$ and covariance matrix 
$$\Sigmabold_{\ybold_{0:N_T}}=\Abold \Qbold^{-1}(\thetabold)\Abold^\top+\sigma_0^2 \Ibold_{n},$$
where $\Qbold(\thetabold)$ is a precision matrix of size $(N_S(N_T+1), N_S(N_T+1))$ depending on the parameters $\thetabold$. For ease of notation, we use $\Qbold$ instead of $\Qbold(\thetabold)$. The log-likelihood is equal to
\begin{equation}
\mathcal{L}(\psibold)=-\frac{n}{2}\log(2\pi)-\frac{1}{2}\log \lvert\Sigmabold_{\ybold_{0:N_T}}\rvert -\frac{1}{2}(\ybold_{0:N_T}-\Lambdabold\bbold)^\top\Sigmabold_{\ybold_{0:N_T}}^{-1}(\ybold_{0:N_T}-\Lambdabold\bbold).
\label{eq:loglike_latent}
\end{equation}

We use the Broyden, Fletcher, Goldfarb, and Shanno optimization algorithm \citep{nocedal06}, that makes use of the second-order derivative of the objective function. The gradients of the log-likelihood function \eqref{eq:loglike_latent} with respect to the different parameters included in $\psibold$ are approximately computed with FD. 

We now propose and detail a  computationally efficient formulation for both the terms of the log-likelihood \eqref{eq:loglike_latent}, namely $-\frac{1}{2}\log \lvert\Sigmabold_{\ybold_{0:N_T}}\rvert$ and $-\frac{1}{2}(\ybold_{0:N_T}-\Lambdabold\bbold)^\top\Sigmabold_{\ybold_{0:N_T}}^{-1}(\ybold_{0:N_T}-\Lambdabold\bbold)$. We first consider the log-determinant. The quadratic form is addressed next.

\begin{proposition}
In the framework outlined above, we have 
\begin{equation}
    \log\lvert\Sigmabold_{\ybold_{0:N_T}}\rvert=n\log\sigma_0^2 -\log\lvert\Qbold\rvert+\log\lvert\Qbold_{\Abold} \rvert,
    \label{eq:logdet_sigmaz}
\end{equation}
where $\Qbold_{\Abold} = \Qbold+\sigma_0^{-2}\Abold^\top\Abold$.
\end{proposition}

\begin{proof}
To compute $\log\lvert\Sigmabold_{\ybold_{0:N_T}}\rvert$, let us consider the augmented matrix 
\begin{equation}
\Sigmabold_c=\begin{pmatrix} \Qbold^{-1} & \Qbold^{-1}\Abold^\top\\ \Abold \Qbold^{-1} & \Sigmabold_{\ybold_{0:N_T}}\end{pmatrix}.
\label{eq:augmented_mat}
\end{equation}
Hence, 
\begin{equation}
 \Qbold_c=\Sigmabold_c^{-1}=\begin{pmatrix} \Qbold + \sigma_0^{-2}\Abold^\top \Abold & - \sigma_0^{-2}\Abold^\top\\  - \sigma_0^{-2} \Abold  & \sigma_0^{-2} \Ibold_n \end{pmatrix}.
\label{eq:augmented_prec_mat}
\end{equation}

\noindent Using block formulas, we have 
$$\log\lvert\Sigmabold_c\rvert=-\log \lvert\Qbold_c\rvert=
-\log\lvert\Qbold\rvert+n\log\sigma_0^2,$$
and
\begin{eqnarray*}
\log\lvert\Sigmabold_c\rvert & = &\log\lvert\Sigmabold_{\ybold_{0:N_T}}\rvert+\log\lvert\Qbold^{-1}-\Qbold^{-1}\Abold^\top\Sigmabold_{\ybold_{0:N_T}}^{-1}\Abold \Qbold^{-1}\rvert\\
& = &\log\lvert\Sigmabold_{\ybold_{0:N_T}}\rvert-\log\lvert\Qbold+\sigma_0^{-2}\Abold^\top\Abold\rvert,
\end{eqnarray*}
where the last equality is a consequence of the Woodbury identity.
This leads to the result.
\end{proof}

\begin{proposition}
The term $\log\lvert\Qbold\rvert$ in Equation \eqref{eq:logdet_sigmaz} can be computed with the computationally efficient formula
\begin{equation}
    \log\lvert \Qbold\rvert = \log\lvert \Sigmabold^{-1}\rvert+(N_T-1)\log\lvert \Fbold^{-1}\rvert,
    \label{eq:logdet_Q}
\end{equation}
with
$$\Fbold^{-1}=\frac{c}{\tilde\tau^2 dt}(\Mbold+\frac{dt}{c}(\Kbold+\Bbold+\Sbold))^\top \Mbold^{-1}\Qbold_S\Mbold^{-1} (\Mbold+\frac{dt}{c}(\Kbold+\Bbold+\Sbold)),$$ 
where $\Qbold_S$ is the precision matrix of the discretized spatial noise $\Zbold_S$ defined in Proposition \ref{prop:discr_adv_diff}. 
\end{proposition}

\begin{proof}
Following \citet{powell2011}, let $\Nbold_N = [\Nbold_{ij}]_{i,j=1}^N$ be an $(nN, nN)$ matrix, which is partitioned into $N^2$ blocks $\Nbold_{ij}$, each of size $(n, n)$. Then the determinant of $\Nbold_N$ is 
$$\lvert \Nbold_N\rvert = \prod_{k=1}^N \lvert \alpha_{kk}^{(N-k)}\rvert,$$
where the $\alpha^{(k)}$ are defined by
\begin{align*}
\alpha_{ij}^{(0)} &= \Nbold_{ij} \\
\alpha_{ij}^{(k+1)} &= \alpha_{ij}^{(k)} - \alpha_{i,N-k}^{(k)} (\alpha_{N-k,N-k}^{(k)})^{-1} \alpha_{N-k,j}^{(k)},\quad k \geq 1.
\end{align*}
$\Qbold$ is a block-matrix organized as $\Nbold_N$. Hence, the formula for $\lvert \Qbold\rvert$ is 
\begin{equation}
\lvert \Qbold\rvert= \lvert \Sigmabold^{-1}\rvert\lvert \Fbold^{-1}\rvert^{(N_T+1)-1}.
\label{eq:det_Q}
\end{equation}
Applying the logarithm, we obtain Equation \eqref{eq:logdet_Q}.
\end{proof}

Note that $\lvert \Fbold^{-1}\rvert$ is now the determinant of a $(N_S, N_S)$ sparse, symmetric and positive definite matrix. The 
log-determinant can be computed through its Cholesky decomposition as
$$\log\lvert \Fbold^{-1}\rvert=2\log(\prod_{i=1}^n F^{-1}_{\text{chol},ii}) = 2\sum_{i=1}^{N_S} \log(F^{-1}_{\text{chol},ii}),$$
since $\Fbold^{-1}_{\text{chol}}$ is a triangular matrix, whose determinant is the product of the diagonal elements.

The term $\log\lvert\Qbold_{\Abold}\rvert=\log\lvert\Qbold+\sigma_0^{-2}\Abold^\top\Abold\rvert$ requires a detailed analysis. The term $\sigma_0^{-2}\Abold^\top\Abold$ is an $(N_S(N_T+1),N_S(N_T+1))$ diagonal block matrix, whose $(N_S,N_S)$ blocks are sparse. The computation of $\log\lvert\Qbold_{\Abold}\rvert$ is not as straightforward as in the case of $\log\lvert\Qbold\rvert$, because there is no way of reducing the computation to purely spatial matrices. Depending on the size $N_S(N_T+1)$, we can either apply a Cholesky decomposition of the $(N_S(N_T+1), N_S(N_T+1))$ matrix $\Qbold_{\Abold}$ or the matrix-free approach proposed in \citet{pereira_desassis_allard_2022}, here sketched. The logarithm function is first approximated by a Chebyschev polynomial $P(\cdot)$ \citep{chebyshev1853}, then the Hutchinson's estimator \citep{hutchinson1990} is used to obtain a stochastic estimate of $\text{tr}[P(\Qbold_{\Abold})]$. The method is detailed in Algorithm 5 in \citet{pereira_desassis_allard_2022}.

\bigskip

Now, concerning the quadratic term of the log-likelihood \eqref{eq:loglike_latent}, we can work with the more convenient expression obtained thanks to the Woodbury formula
$$\Sigmabold_{\ybold_{0:N_T}}^{-1}=\sigma_0^{-2}\Ibold_{n}-\sigma_0^{-4}\Abold\Qbold_{\Abold}^{-1}\Abold^\top.$$
Hence
\begin{equation}
\begin{split}
(\ybold_{0:N_T}-\Lambdabold\bbold)^\top \Sigmabold_{\ybold_{0:N_T}}^{-1} &(\ybold_{0:N_T}  -  \Lambdabold\bbold) 
 = \sigma_0^{-2}(\ybold_{0:N_T}-\Lambdabold\bbold)^\top \Ibold_{n}(\ybold_{0:N_T}-\Lambdabold\bbold)\\
&-\sigma_0^{-4}(\ybold_{0:N_T}-\Lambdabold\bbold)^\top \Abold\Qbold_{\Abold}^{-1}\Abold^\top(\ybold_{0:N_T}-\Lambdabold\bbold).
\end{split}
\label{eq:quadratic_form}
\end{equation}
The second term of Equation \eqref{eq:quadratic_form} can be computed either by Cholesky decomposition or using the conjugate gradient (CG) method. This latter method solves $\Qbold_{\Abold} \vbold=\wbold$ with respect to $\vbold$ and computes $\vbold_{\text{sol}} = \wbold^\top \vbold$, with $\wbold=\Abold^\top(\ybold_{0:N_T}-\Lambdabold\bbold)$. In this case, it is useful to find a good preconditioner for the matrix $\Qbold_{\Abold}$ to ensure fast convergence of the CG method. We found that a temporal block Gauss-Seidel preconditioner \citep[Chapter 3]{young71} was a good choice in this case. A detailed presentation of the CG method is available in \citet{pereira_desassis_allard_2022}. 
\medskip


\subsection{Prediction by Kriging}
\label{sec:kriging}

Under a Gaussian assumption, optimal prediction is the conditional expectation, also known in the geostatistics literature as kriging.
We detail here two prediction settings: space-time interpolation and temporal extrapolation. 

In the space-time interpolation setting, the spatio-temporal vector  $\xbold_{0:N_T}$ is predicted on the entire spatial mesh during the time window $\left[0,T\right]$, i.e. on ${\cal T}'$, using the data $\ybold_{0:N_T}$  defined in Equation \eqref{eq:stat_model}. The kriging predictor is directly read from Equation \eqref{eq:augmented_prec_mat}:

\begin{equation}
\xbold_{0:N_T}^\star =  \mathds{E}(\xbold_{0:N_T}\mid \ybold_{0:N_T}) 
= \sigma_0^{-2} \Qbold_{\Abold}^{-1}\Abold^\top (\ybold_{0:N_T}- \Lambdabold\widehat{\bbold}).
    \label{eq:kriging}
\end{equation}

The computation of \eqref{eq:kriging} requires the inversion of $\Qbold_{\Abold}$, as detailed in Section \ref{sec:estimation_scattered}. The conditional variance, also called kriging variance, is 
$$
\text{Var}(\xbold_{0:N_T} \lvert \ybold_{0:N_T}) = \Qbold_{\Abold}^{-1}.
$$

The computation of the diagonal of an inverse matrix is not straightforward when only the Cholesky decomposition of the matrix is available. Among the existing methods there is the Takahashi recursive algorithm described in \citet{takahashi73} and \citet{erisman75}. Another way of computing the kriging variance is through conditional simulations, as detailed in Section \ref{sec:conditional_simu}.

In the temporal extrapolation setting, the vector $\xbold^{(N_T+1)}$ is predicted at time step $(N_T+1)$ on $\cal T$ using all the data available until time $T$, i.e. from $\ybold_{0:N_T}$. Following Equation \eqref{eq:timestep}, we have
\begin{equation}
\xbold^{(N_T+1)} = \Dbold\xbold^{(N_T)} + \Ebold \zbold^{(N_T+1)},
\label{eq:dynamics}
\end{equation}
where $\zbold^{(N_T+1)}$ is a standardized Gaussian vector, and $\Dbold$ and $\Ebold$ are defined in Proposition \ref{prop:grmf}. The kriging predictor $\xbold^{\star(N_T+1)}$ is
\begin{eqnarray}
\xbold^{\star(N_T+1)} & = & \mathds{E}(\xbold^{(N_T+1)}\mid \ybold_{0:N_T}) = \Dbold \mathds{E}(\xbold^{(N_T)}\mid \ybold_{0:N_T}) = \Dbold \xbold^{\star(N_T)},
\label{eq:kriging_t1}
\end{eqnarray}
where $\xbold^{\star(N_T)}$ is extracted from $\xbold_{0:N_T}^\star$. The same procedure can be iterated to predict $\xbold$ at further time steps.

\subsection{Conditional simulations}
\label{sec:conditional_simu}

To perform a conditional simulation, we use the conditional kriging paradigm presented below. This approach relies on the fact that  kriging predictors and kriging residuals are uncorrelated (independent under Gaussian assumption, see \citet[Chapter 7]{chiles99}). First, a non-conditional simulation $\xbold^{(NC)}_{0:N_T}$ is performed on the spatio-temporal grid ${\cal T}'$. From this simulation, kriging residuals
$$\rbold_{0:N_T} = 
\mathds{E}\left(\xbold_{0:N_T} \mid \Abold \xbold^{(NC)}_{0:N_T}\right) - \xbold^{(NC)}_{0:N_T}
$$
are computed over the entire spatio-temporal grid $\mathcal{T}'$. The conditional expectation is computed using the method presented in the previous section. In a second step, these independently generated residuals are added to the usual kriging of the data to get the conditional simulation 
$$\xbold^{(C)}_{0:N_T} = \xbold^{*}_{0:N_T} + \rbold_{0:N_T}.$$
Conditional simulations at further time steps are obtained by iteratively computing   
$\xbold^{(C)}_{N_T+k}$  using the propagation equation \eqref{eq:dynamics} with $k\geq 1$. Multiple independent realizations of conditional simulations can then be used to compute estimates of conditional variances or other quantities, such as probability maps of threshold exceedance.


\subsection{Simulation study}

We report here some results regarding the estimation of the parameters $\thetabold^\top = [\kappa, \gamma_1, \gamma_2, c, \tau]$ for 50 independent realisations of a spatio-temporal model simulated with the SPDE \eqref{eq:adv_diff}. We set $\Hbold=\Ibold$, $\alpha=1$ and $\alpha_S=2$.  The spatial domain is the $[0,30]^2$ square with a grid triangulation of $N_S=900$ spatial points. The time window is $[1,10]$ with unit time step and $N_T+1=10$. The $n_S=100$ observations are randomly located into the spatial domain and their position do not change during the $N_T$ time steps (hence $n=1000$). Since the sizes of both the dataset and the spatio-temporal mesh are reasonable, we report the estimations computed with both the Cholesky decomposition approach and the matrix-free approach. 

As initial values, we use estimated values obtained from the variograms of the spatial and temporal traces of the process. Specifically, the initial value for $\kappa$ is the estimated scale parameter of a Mat\'ern covariance function with smoothness parameter $\nu = \alpha +\alpha_S - 1=2$ considering independent temporal repetitions, the initial value for $c$ is deduced from the estimated parameter of $n_S$ independent repetitions of AR(1) processes of length $(N_T+1)$ and $\tau^2$ is computed from Equation \eqref{eq:spat_trace} with $\sigma^2$ being the empirical variance computed on the data. Finally, the initial value for $\gammabold$ is the null vector. The parameters for the matrix-free approach  are set to the following: the order of the Chebychev polynomial to approximate the logarithm is set to 30 and the number of terms in the sum of the  Hutchinson's estimator is set to 10. The results are reported in Table \ref{tab:ml_spatio_temp}. They show that all parameters are accurately estimated with both approaches. In almost all cases, the true value of the parameter is within the mean $\pm$ 2 standard deviations interval. We remark how the matrix-free approach takes more time to estimate the parameters. This is due to the iterative computations, that increase the computational time. However, we know that the benefit of the matrix-free approach is the possibility of applying it to much larger spatio-temporal meshes, where the Cholesky decomposition cannot be applied at all.

\begin{table}[hbt]
\begin{adjustbox}{width=\columnwidth,center}
\begin{tabular}{c|cccccc}
\toprule
method & $\kappa$ & $\gamma_1$ & $\gamma_2$ & $c$ & $\tau$ & average time $(s)$\\
\midrule
    correct & 0.5 & 2 & 2 & 1 & 1 & \\
	Cholesky & 0.610 (0.047) & 2.354 (0.515) & 2.325 (0.421) & 1.037 (0.218) & 1.072 (0.040) & 194 \\
    Matrix-free & 0.483 (0.029) & 1.904 (0.178) & 1.906 (0.147) & 1.027 (0.046) &  0.998 (0.024) & 960\\
	\midrule
    correct & 0.7 & 1 & -1 & 2 & 0.5 & \\
	Cholesky & 0.695 (0.067) & 1.056 (0.659) & -1.134 (0.631) & 2.090 (0.352) & 0.503 (0.022) & 172\\
    Matrix-free & 0.669 (0.052) & 0.967 (0.123) &  -1.164 (0.142) & 1.954 (0.118) & 0.485 (0.013) & 863 \\
\bottomrule
\end{tabular}
\end{adjustbox}
\caption{Mean (and standard deviation) of ML estimates $\hat{\thetabold}^\top = [\hat{\kappa}, \hat{\gamma}_1, \hat{\gamma}_2, \hat{c}, \hat{\tau}]$ over 50 simulations for two different subsets of advection-diffusion model parameters and two different estimation approaches (Cholesky decomposition and matrix-free approach).}

\label{tab:ml_spatio_temp}
\end{table}


\section{Application to a solar radiation dataset}
\label{sec:application}

With the constantly increasing installation of photovoltaic (PV) power and its volatility due to weather conditions, characterizing short-term variability of generated solar power (from the minute resolution to the 15-minute resolution) is important for the integration of PV systems into the electrical grid,  for balancing supply and demand \citep{kreuwel2020analysis}.  Fluctuations in solar production can have a significant impact on grid stability, and accurate prediction allows for planning necessary adjustments, such as modulating the production of other energy sources, to avoid service interruptions and optimize resource utilization.

The approach detailed in the previous sections is now applied to a solar radiation dataset for which experts agree on the presence of advection due to Western prevailing winds transporting clouds from one side of the domain to the other. The HOPE campaign \citep{macke2017} recorded  Global Horizontal Irradiance (GHI) (also called  SSI, Surface Solar Irradiance) over a $10 \times 16 \ km^2$ region in West Germany near the city of J\"ulich  from April 2 to July 2, 2013. The sensors were located at 99 stations located as pictured in Figure \ref{fig:locations} and GHI was recorded every 15 seconds. A detailed description of the campaign can be found in \citet{macke2017}.

\begin{figure}[htb]
\centering
\includegraphics[width=0.35\textwidth]{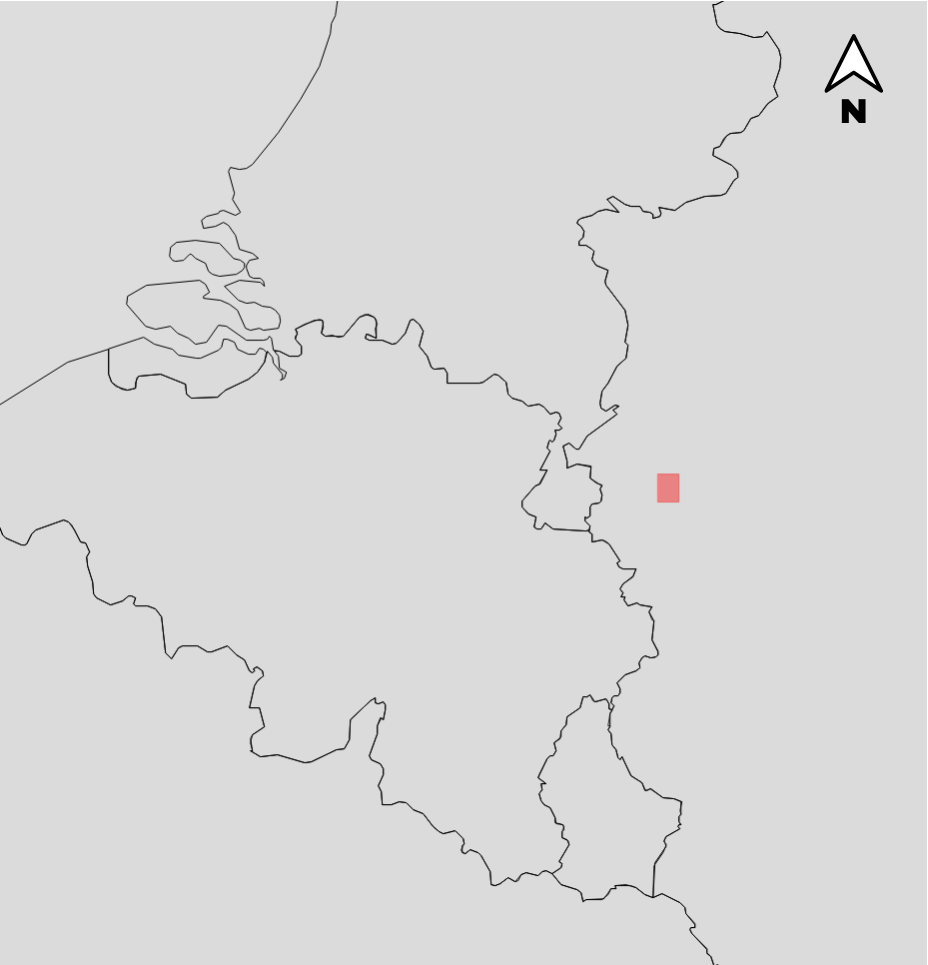}
\includegraphics[width=0.35\textwidth]{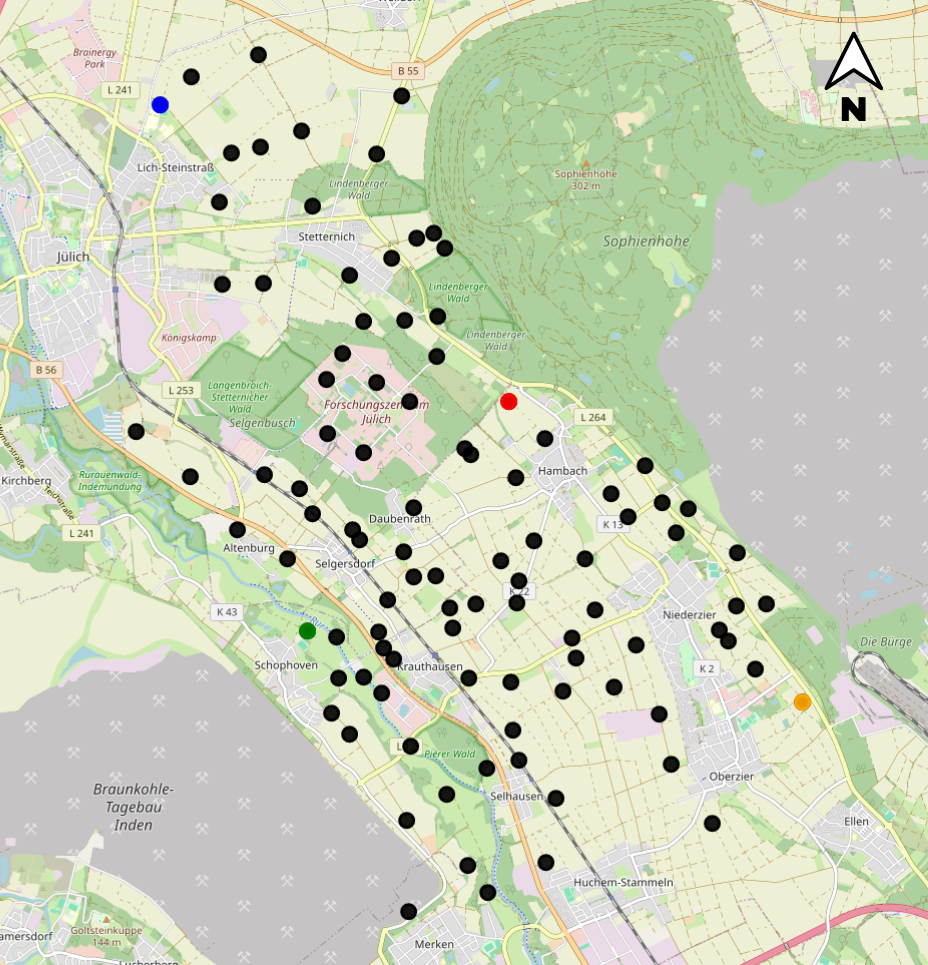}
\caption{Left: spatial domain of study (red square). Right: zoom on the spatial domain along with measuring stations.}
\label{fig:locations}
\end{figure}

The dataset was cleaned for outlying values and non-operating sensors, and the temporal resolution was reduced from 15 seconds to 1 minute. Figure \ref{fig:4_stations} (left panel) shows  GHI as a function of time (in minute, during a full day -- the 28th of May 2013) at 4 different stations. These stations, represented in color in Figure \ref{fig:locations},  are located at the border of the domain, far from each other. The GHI starts close to 0, increases after sunrise, peaks at midday and tends to 0 at sunset. The maximal theoretical amount of irradiance reaching the sensor follows an ideal concave curve. The divergence between the measured irradiance and the optimal curve can be slight or important, depending on the presence of clouds. One can see on this example that the evolution among the 4 stations is similar, with variations accounting for spatio-temporal variations of the clouds.

\begin{figure}
\centering
\includegraphics[width=0.45\textwidth]{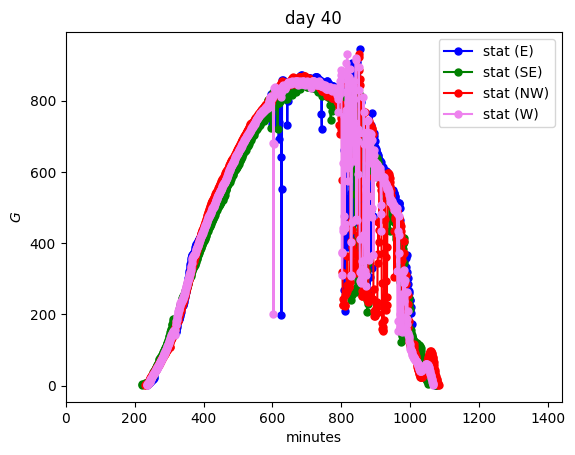}
\includegraphics[width=0.45\textwidth]{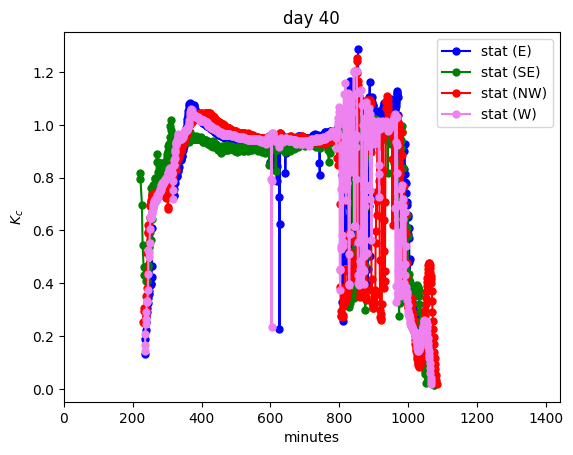}
\caption{GHI $G$ and Clear Sky Index $K_c$ for 4 different stations on the 28th of May 2013}
\label{fig:4_stations}
\end{figure}

A first preprocessing was made in order to stationarize the phenomenon. \citet{oumbe2014} showed that the solar irradiance  at ground level, GHI (denoted $G$ for short from now on), computed by a radiative transfer model can be approximated by the product of the irradiance under clear atmosphere (called Clear Sky GHI, or $G_c$) and a modification factor due to cloud properties and ground albedo only (Clear Sky Index, or $K_c$, \citet{beyer96}):
\begin{equation}
    G \simeq G_c K_c.
\label{eq:ghi_general} 
\end{equation}
The error made using this approximation depends mostly on the solar zenith angle, the ground albedo and the cloud optical depth. In most cases, the maximum errors (95th percentile) on global and direct surface irradiances are less than 15 $Wm^{-2}$ and less than 2 to 5 \% in relative value, as recommended by the World Meteorological Organization for high-quality measurements of the solar irradiance \citep{oumbe2014}. Practically, it means that a model for fast calculation of surface solar irradiance may be separated into two distinct and independent models:  a deterministic model for $G$, under clear-sky conditions, as computed according to \citet{gschwind2019}, considered as known in this study, and a model for $K_c$, which accounts for cloud influence on the downwelling radiation and is expected to change in time and space. $K_c$ is modeled as a random spatio-temporal process and will be the subject of our analysis. Figure  \ref{fig:4_stations} (right panel) shows the variable $K_c$ corresponding to the variable $G$ shown on the left panel. In general, $K_c$ lies between 0 and 1, but in rare occasions, values above 1 can be observed. This phenomenon is called \textit{overshooting} \citep{schade2007} and is due to light reflection by surrounding clouds.

A time window of 20 minutes around 4 p.m. on May 28, 2013 is extracted with observations every minute at the 73 stations with well recorded values. The histogram and time series of the data are shown in Figure \ref{fig:hist}. Parameters are estimated on this 20-minute window using the method described in Section \ref{sec:estimation_scattered}. The spatio-temporal grid contains $N_T+1 = 20$ one-minute time steps, from $t=1$ to $t=20$ and $N_S = 900$ regular spatial mesh points.

\begin{figure}
\centering
\includegraphics[width=0.4\textwidth]{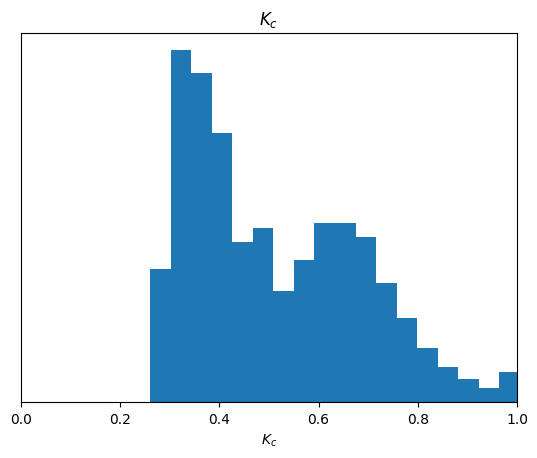}
\includegraphics[width=0.43\textwidth]{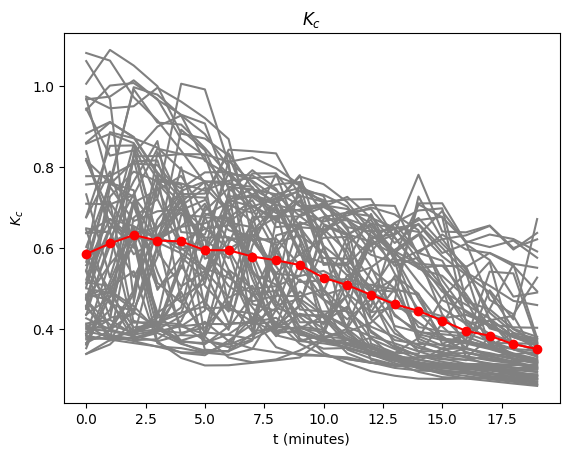}
\caption{Left: Histogram of $K_c$ over 20 time steps. Right: Time series of $K_c$ for all the stations over 20 time steps, along with the mean value for each time step (in red)}
\label{fig:hist}
\end{figure}

\subsection{Estimation and prediction}
\label{sec:results}

Six different models are fitted to the data and used for  prediction: 3 models with advection (called “adv-diff") and 3 models without advection (called “diff") obtained by setting $\gammabold=\0bold$. Both groups contain the three following sub-models: (i) a model with diffusion included only in the stochastic forcing term, with a Mat\'ern spatial trace with $\nu = 1$; (ii) a nonseparable model that does not have a Mat\'ern spatial trace, but a generalized covariance function instead (the spatial covariance function exists only if $\alpha+\alpha_S>d/2$, where $d=2$ in this case, see Proposition \ref{prop:spatial_trace}); in this case, the SPDE generates a process which only has meaning as a random measure, and cannot strictly be interpreted at individual
locations; (iii) a nonseparable model with a Mat\'ern spatial trace with $\nu = 2$. In the general model of Equation \eqref{eq:adv_diff} they correspond respectively to $(\alpha,\alpha_S)=(0,2),(1,0),(1,2)$. The parameters of the SPDE are estimated for each model separately. The results are reported in Table \ref{tab:estimates}. 

The log-likelihoods of the models that include advection are  within a range of variations of 10 log-likelihood units and are between 34 to 80 units larger than those with diffusion only. As a point of comparison, if all spatio-temporal dependencies were ignored, the BIC penalization for the advection parameters would be $2 \ln(1460)\simeq 14.5$. These results indicate strong evidence in favor of models with advection, but no significant differences among them. The parameters vary substantially from one model to the other, but it must be remembered that, when considered independently, their physical interpretation is model dependent. Some combinations are interpretable however. For example, following Proposition \ref{prop:spatial_trace}, the overall variance is equal to $(8\pi)^{-1}\tau^2 \kappa^{-2}$ (or $(8\pi)^{-1}\tau^2\kappa^{-2}\lvert\Ibold + h\norm{\gammabold}^{-1}\gammabold\gammabold^\top \rvert^{-1/2}$ in the stabilized case) when $\alpha_{\text{tot}} = \alpha+\alpha_S=2$ and it is equal to $(16\pi)^{-1}\tau^2\kappa^{-4}$ (or $(16\pi)^{-1}\tau^2 \kappa^{-4}\lvert\Ibold + h\norm{\gammabold}^{-1}\gammabold\gammabold^\top \rvert^{-1/2}$) when $\alpha_{\text{tot}}=3$.
Accordingly, the estimated standard deviations for models (1), (3), (4) and (6) are equal to 0.160, 0.119, 0.174 and 0.199 respectively, with the experimental standard deviation being equal to 0.184.
For the same models, the practical ranges computed as $\sqrt{8\nu}/\kappa$ \citep{lindgren2011} are equal to 1.915, 3.079, 2.281 and 4.255 respectively. Notice that among pairs of models that differ by the presence or absence of advection, the estimated range is larger for those without advection in an attempt to account for the larger correlation distance due to transport.

\begin{table}[hbt]
\begin{adjustbox}{width=\columnwidth,center}
\begin{tabular}{cccc|c|ccccccc}
\toprule
& Model & $\alpha$ & $\alpha_S$ & log-likelihood & $\hat{\kappa}$ & $\hat{\gamma}_x$ & $\hat{\gamma}_y$ & $\hat{c}$ & $\hat{\tau}$ & $\hat{\sigma}_0$ & $\hat{b}$\\
\midrule
(1) & adv-diff & 0 & 2 & 2587 & 1.477 & 9.642 & -5.382 & 11.659 & 2.254 & 0.052 & 0.570\\
(2) & adv-diff & 1 & 0 & 2577 & 0.237 & 4.718 & -0.928 & 9.315 & 0.458 & 0.045 & 0.598\\
(3) & adv-diff & 1 & 2 & 2579 & 1.299 & 17.325 & -8.442 & 41.017 & 3.072 & 0.058 & 0.574\\
(4) & diff & 0 & 2 & 2507 & 1.240 & 0 & 0 & 12.558 & 1.081 & 0.059 & 0.569\\
(5) & diff & 1 & 0 & 2545 & 0.246 & 0 & 0 & 6.594 & 0.436 & 0.047 & 0.577\\
(6) & diff & 1 & 2 & 2512 & 0.940 & 0 & 0 & 34.607 & 1.248 & 0.064 & 0.580\\
\botrule
\end{tabular}
\end{adjustbox}
\caption{Estimated parameters and log-likelihood for 6 different models from all data on a 20-minute window}
\label{tab:estimates}
\end{table}

With the objective of improving the predictions at the 15-minute resolution needed for the electrical grid management, we then perform prediction with two different validation settings containing 80\% of conditioning data and 20\% of validation data. In the first case (called “Uniform") the validation locations are uniformly randomly selected. In the second case (called “South-East") the validation locations are located downwind (i.e. South-East) with respect to the estimated advection direction. See Figure \ref{fig:stat} for a representation of the validation settings.

\begin{figure}
\centering
\includegraphics[width=0.32\textwidth]{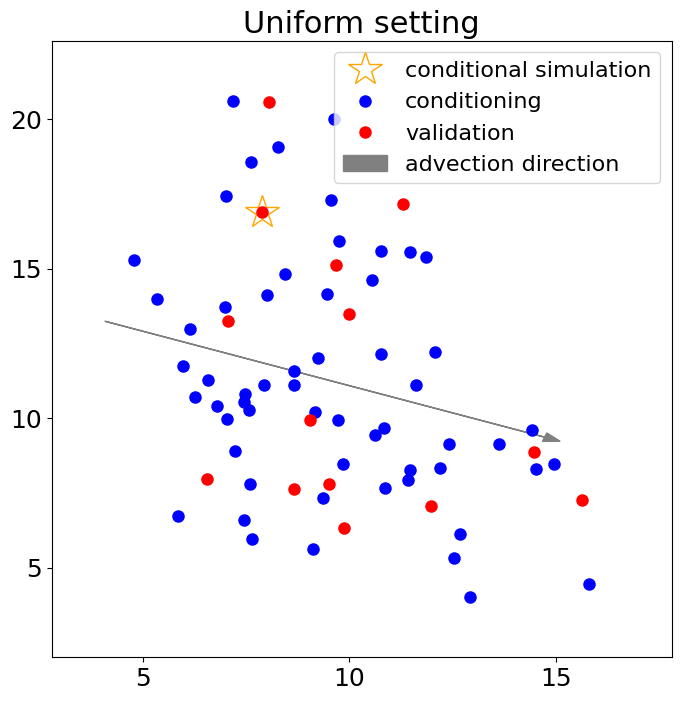}
\includegraphics[width=0.32\textwidth]{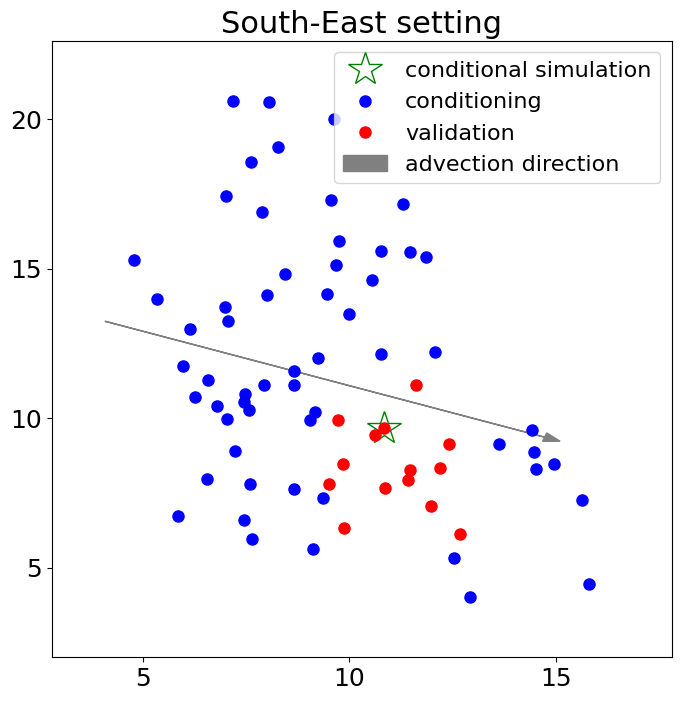}
\caption{Validation settings: Uniform (left) and South-East (right)}
\label{fig:stat}
\end{figure}

\begin{figure}
\centering
\includegraphics[width=0.99\textwidth]{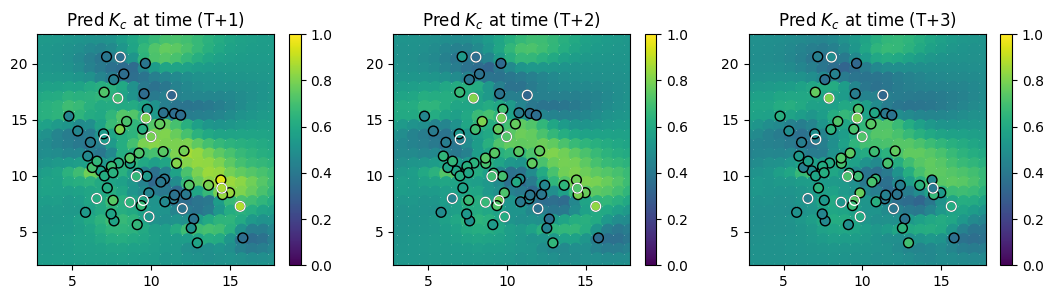}
\caption{Predictions of $K_c$ at $(T+1)$, $(T+2)$ and $(T+3)$ with model (3) (“adv-diff" with $\alpha=1$ and $\alpha_S=2$). The black contoured dots are the conditioning locations and the white contoured dots are the validation locations for the Uniform setting}
\label{fig:error}
\end{figure}

Recall that a time window containing 20 unit time-steps, from $t=1$ to $t=20$, has been selected. For each validation setting and considering each time the end of the time window from $T=11$ to $T=20$, three prediction configurations using conditioning data from time $(T-9)$ to $T$ are computed and compared to the real values, allowing us to compute a root mean square error (RMSE) validation score. First, the kriging is performed spatially only (hereafter referred to as “S" kriging). Second, a temporal extrapolation is computed at the conditioning locations at time horizons $(T+1)$, $(T+2)$ and $(T+3)$  (“T1", “T2", “T3" kriging). Third, the spatio-temporal prediction is computed at the validation locations at time horizons $(T+1)$, $(T+2)$ and $(T+3)$ (“ST1", “ST2", “ST3" kriging). 
We thus have a total of 6 models $\times$ 2 validation settings $\times$ 3 prediction configurations. RMSEs are averaged over the 10 repetitions. Results are reported in Table \ref{tab:scores} and those of the ST prediction configuration are also shown in Figure \ref{fig:rmse}.

\begin{table}[h!]
\begin{adjustbox}{width=0.95\columnwidth,center}
\begin{tabular}{cccc|ccc}
\toprule
& &  &  & \multicolumn{3}{c}{\textbf{Uniform}} \\
\midrule
& Model & $\alpha$ & $\alpha_S$ & \multicolumn{3}{c}{S (min,max)}\\
\midrule
(1) & adv-diff & 0 & 2 & \multicolumn{3}{c}{\textbf{0.088 (0.052,0.127)}} \\
(2) & adv-diff & 1 & 0 & \multicolumn{3}{c}{0.103 (0.064,0.142)} \\
(3) & adv-diff & 1 & 2 & \multicolumn{3}{c}{0.102 (0.062,0.134)} \\
(4) & diff & 0 & 2 & \multicolumn{3}{c}{0.119 (0.074,0.140)}\\
(5) & diff & 1 & 0 & \multicolumn{3}{c}{0.094 (0.060,0.131)}\\
(6) & diff & 1 & 2 & \multicolumn{3}{c}{0.110 (0.066,0.132)}\\
\midrule
& Model & $\alpha$ & $\alpha_S$ & T1 (min,max) & T2 (min,max) & T3 (min,max)\\
\midrule
(1) & adv-diff & 0 & 2 & 0.095 (0.067,0.120) & 0.146 (0.111,0.186) & 0.181 (0.131,0.236)\\
(2) & adv-diff & 1 & 0 & \textbf{0.071 (0.046,0.090)} & \textbf{0.093 (0.054,0.124)} & \textbf{0.102 (0.060,0.143)}\\
(3) & adv-diff & 1 & 2 & 0.072 (0.046,0.093) & 0.095 (0.054,0.127) & 0.104 (0.055,0.144)\\
(4) & diff & 0 & 2 & 0.094 (0.058,0.123) & 0.137 (0.091,0.181) & 0.166 (0.102,0.231)\\
(5) & diff & 1 & 0 &  0.079 (0.058,0.098) & 0.108 (0.082,0.135) & 0.124 (0.099,0.158)\\
(6) & diff & 1 & 2 & 0.083 (0.054,0.108) & 0.110 (0.077,0.149) & 0.125 (0.085,0.180)\\
\midrule
& Model & $\alpha$ & $\alpha_S$ & ST1 (min,max) & ST2 (min,max) & ST3 (min,max)\\
\midrule
(1) & adv-diff & 0 & 2 & 0.105 (0.067,0.144) & 0.147 (0.106,0.193) & 0.179 (0.140,0.231)\\
(2) & adv-diff & 1 & 0 & 0.091 (0.052,0.131) & 0.103 (0.062,0.161) & 0.110 (0.071,0.165)\\
(3) & adv-diff & 1 & 2 & \textbf{0.085 (0.050,0.127)} & \textbf{0.094 (0.058,0.142)} & \textbf{0.100 (0.061,0.157)}\\
(4) & diff & 0 & 2 & 0.123 (0.072,0.186) & 0.150 (0.081,0.237) & 0.170 (0.116,0.257)\\
(5) & diff & 1 & 0 & 0.104 (0.070,0.153) & 0.122 (0.095,0.181) & 0.131 (0.095,0.187)\\
(6) & diff & 1 & 2 & 0.108 (0.073,0.153) & 0.126 (0.085,0.188) & 0.134 (0.082,0.199)\\
\midrule
& &  &  & \multicolumn{3}{c}{\textbf{South-East}}\\
\midrule
& Model & $\alpha$ & $\alpha_S$ & \multicolumn{3}{c}{S (min,max)}\\
\midrule
(1) & adv-diff & 0 & 2 & \multicolumn{3}{c}{\textbf{0.103 (0.051,0.138)}}\\
(2) & adv-diff & 1 & 0 & \multicolumn{3}{c}{0.105 (0.045,0.158)}\\
(3) & adv-diff & 1 & 2 & \multicolumn{3}{c}{0.109 (0.054,0.149)}\\
(4) & diff & 0 & 2 & \multicolumn{3}{c}{0.134 (0.092,0.181)}\\
(5) & diff & 1 & 0 & \multicolumn{3}{c}{0.136 (0.067,0.187)}\\
(6) & diff & 1 & 2 & \multicolumn{3}{c}{0.140 (0.085,0.192)}\\
\midrule
& Model & $\alpha$ & $\alpha_S$ & T1 (min,max) & T2 (min,max) & T3 (min,max)\\
\midrule
(1) & adv-diff & 0 & 2 & 0.095 (0.065,0.122) & 0.142 (0.106,0.185) & 0.172 (0.121,0.228)\\
(2) & adv-diff & 1 & 0 & 0.074 (0.045,0.099) & 0.097 (0.065,0.128) & 0.109 (0.069,0.148)\\
(3) & adv-diff & 1 & 2 & \textbf{0.074 (0.049,0.097)} & \textbf{0.096 (0.062,0.122)} & \textbf{0.106 (0.061,0.139)}\\
(4) & diff & 0 & 2 & 0.090 (0.063,0.116) & 0.128 (0.097,0.167) & 0.154 (0.113,0.209)\\
(5) & diff & 1 & 0 & 0.081 (0.057,0.105) & 0.111 (0.090,0.147) & 0.128 (0.100,0.169)\\
(6) & diff & 1 & 2 & 0.084 (0.056,0.109) & 0.109 (0.086,0.152) & 0.123 (0.091,0.176)\\
\midrule
& Model & $\alpha$ & $\alpha_S$ & ST1 (min,max) & ST2 (min,max) & ST3 (min,max)\\
\midrule
(1) & adv-diff & 0 & 2 & 0.102 (0.081,0.121) & 0.158 (0.130,0.180) & 0.210 (0.160,0.241)\\
(2) & adv-diff & 1 & 0 & \textbf{0.079 (0.044,0.119)} & \textbf{0.075 (0.039,0.127)} & \textbf{0.070 (0.038,0.132)}\\
(3) & adv-diff & 1 & 2 & 0.090 (0.052,0.121) & 0.099 (0.056,0.149) & 0.109 (0.057,0.172)\\
(4) & diff & 0 & 2 & 0.128 (0.092,0.157) & 0.165 (0.112,0.195) & 0.199 (0.116,0.236)\\
(5) & diff & 1 & 0 & 0.107 (0.050,0.201) & 0.100 (0.042,0.223) & 0.094 (0.0380,0.217) \\
(6) & diff & 1 & 2 & 0.114 (0.060,0.204) & 0.111 (0.058,0.229) & 0.108 (0.059,0.223)\\
\botrule
\end{tabular}
\end{adjustbox}
\caption{Averaged RMSE computed at 10 successive time steps for 6 different models, 2 validation settings (Uniform and South-East) and 3 prediction configurations (S, T and ST); see text for details. In each case, the best score among the models is in bold font}
\label{tab:scores}
\end{table}

\begin{figure}
\centering
\includegraphics[width=0.49\textwidth]{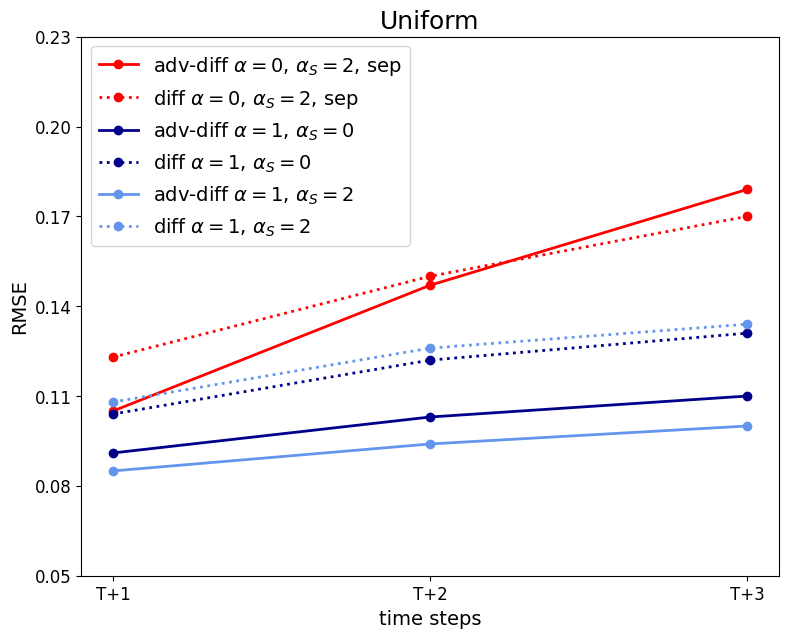}
\includegraphics[width=0.49\textwidth]{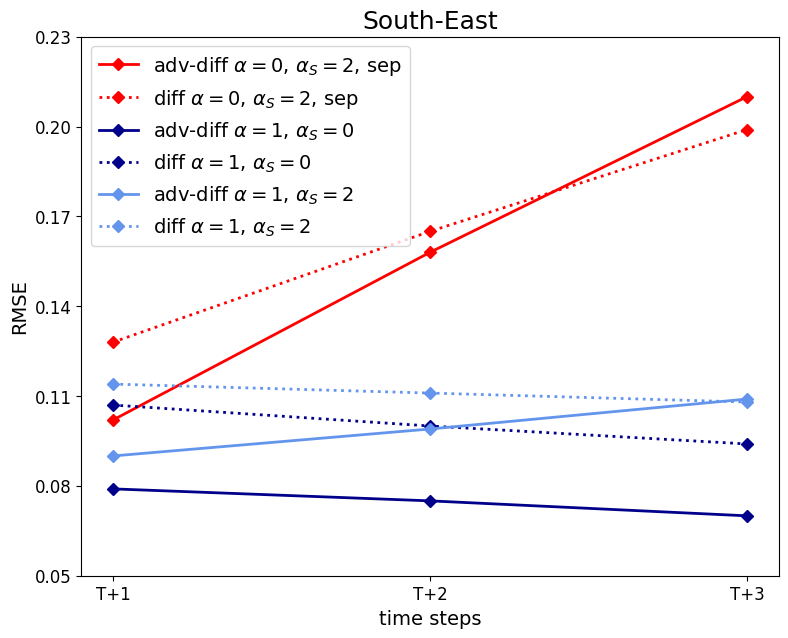}
\caption{Averaged RMSE for ST prediction configuration for 6 different models and 2 validation settings: Uniform (left) and South-East (right)}
\label{fig:rmse}
\end{figure}

For all tested validation settings and prediction configurations, the models with advection show better RMSE scores than models without advection. This result is a confirmation of the results already observed on log-likelihoods. Models with advection have similar prediction scores in the prediction configurations S and T, model (1) having slightly better performances in the configuration S. In the T and ST configurations, models (2) and (3) have in general quite similar RMSEs, except in the South-East setting with ST configuration where model (2) is clearly the best model. In this case, prediction is made in a space-time domain lying downstream with respect to the advection. It is thus expected that the model best representing the underlying physics should lead to the best prediction performances. 


An example of prediction maps on $\mathcal{T}$ at time horizons $(T+1)$, $(T+2)$ and $(T+3)$ is reported in Figure \ref{fig:error}, along with the observed values (black contoured dots). The white contoured dots are the locations used for validation in the Uniform setting.

\subsection{Conditional simulations}

Figure \ref{fig:cond_simu} shows 100 conditional simulations of $K_c$ computed at time $T=11$ and horizons $(T+1), (T+2), \dots, (T+6)$ with the advection-diffusion model (3). Two validation stations have been selected: one in the North-West part of the domain (the orange star in the left panel of Figure \ref{fig:stat}) and one in the South-East part of the domain (the green star in the right panel of Figure \ref{fig:stat}). Given that there is an advection from NW to SE, it is therefore expected that the advection-diffusion model should be able to transport the information in that direction. The mean of the 100 simulations and the envelopes corresponding to twice the pointwise standard deviation have also been represented, along with the true values. As expected, most of the conditional simulations lie within the envelopes in both cases and at all time horizons. The remarkable result is that the variance of the conditional simulations at the green station is smaller than that at the orange one at every time step, especially when the time horizon increases. This is due to the advection term in model (3), able to propagate information from North-West to South-East. 

\begin{figure}
\centering
\includegraphics[width=0.49\textwidth]{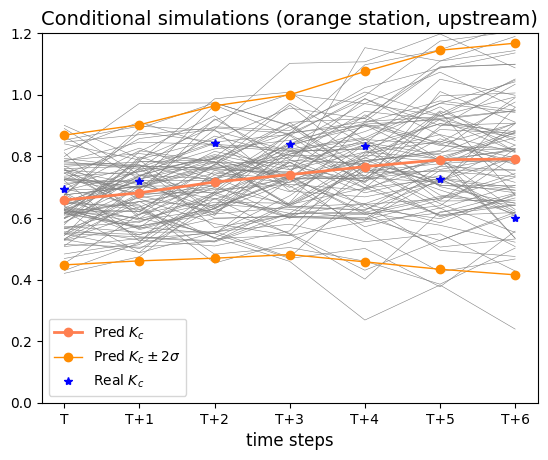}
\includegraphics[width=0.49\textwidth]{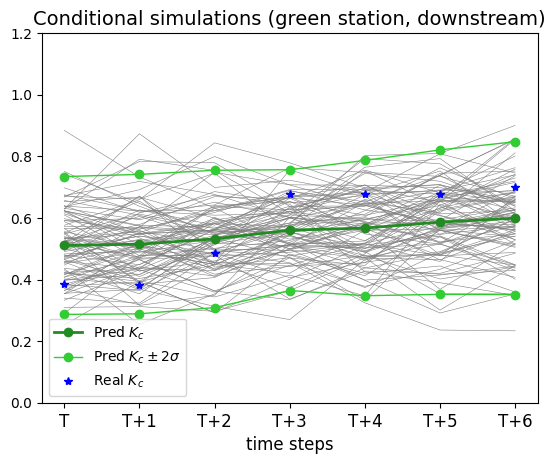}
\caption{Real $K_c$, mean of conditional simulations of $K_c$ and $\pm2\sigma$ envelope at time horizons $T,(T+1),(T+2),\dots,(T+6)$. Left: orange station in the North-West part of the domain. Right: green station in the South-East part of the domain}
\label{fig:cond_simu}
\end{figure}


\section{Conclusion}
\label{sec:discussion}

The spatio-temporal SPDE approach based on advection-diffusion equations proposed in this work combines elements of physics, numerical analysis and statistics. It can be seen as a first step toward \textit{physics informed geostatistics}, which introduces physical dynamics into a  statistical model, accounting for possible hidden structures governing the evolution of the spatio-temporal phenomenon.  
The different terms of the SPDE (advection, diffusion) directly influence the spatio-temporal dependencies of the process, by controlling its variability in space and time. Compared to spatio-temporal models built on covariance functions such as the Gneiting class \citep{gneiting2002}, we gain in interpretability since the parameters of the model can be linked to the physical coefficients of SPDEs.

We showed that it is possible to build an accurate space-time approximations of the process driven by the advection-diffusion SPDE using a combination of FEM in space and implicit Euler scheme in time. It leads to sparse structured linear systems. We obtained
promising results for the estimation and for the prediction of processes both in terms of precision and speed. When the size of the dataset is moderate, direct matrix implementation is possible. We showed how matrix-free methods can be implemented in order to obtain scalable computations even for very large datasets. The application to the solar radiation dataset demonstrated that the nonseparable advection-diffusion model exhibited the best prediction performances on a phenomenon that is certainly governed by advection and diffusion processes.  Nonetheless, further work is necessary to better assess the prediction accuracy and the computational complexity. Applications to larger and more complex datasets, in particular using the matrix-free approach, will be considered.

Further work is also necessary to compare the proposed approach to competing ones. In the spirit of \citet{heaton19}, which focused on spatial models, a comparison aimed at spatio-temporal processes showing dominating advection would be of great interest. For example, \citet{okasaki2022} considered spatio-temporal SPDEs without advection or with an advection that does not dominate diffusion. Moreover, the moderate size of the dataset considered allowed the use of direct matrix computations whereas, as discussed above, our approach is scalable. Comparison to models expressing the advection in a Lagrangian framework \citep{ailliot2011space,benoit2018stochastic,salvana2021} should also be performed. 

A maximum likelihood procedure was implemented. As a follow-up work, it would be interesting to implement this space-time model as part of a Bayesian hierarchical construction, possibly within the INLA/SPDE framework \citep{rue2009,krainski2018advanced}, which already propose separable spatio-temporal models and will probably soon include the diffusion spatio-temporal SPDE model proposed in \citet{lindgren2022diffusion}. The Bayesian framework would enable the assessment of estimation and prediction uncertainty. Extending the INLA approach to deal with advection-diffusion models is left for future work. Different parameter estimation methods could also be investigated, such as “online" methods, which recursively estimate the unknown model parameters based on the continuous stream
of observations \citep{sharrock2022}. 

One of the main advantages of the SPDE formulation is that it is easy to generalize to nonstationary settings. Nonstationary fields can be defined by letting the parameters $\kappa(t,\sbold)$ and $\gammabold(t,\sbold)$ be space-time-dependent. This generalization implies only minimal changes to the method used in the stationary case concerning the simulation, but needs more work for estimation and prediction, since the maximum likelihood approach becomes much more expensive. We can also incorporate models of spatially varying anisotropy by modifying the general operator $\nabla \cdot \Hbold(t,\sbold)\nabla X(t,\sbold)$ with a nonstationary anisotropic matrix $\Hbold(t,\sbold)$. The introduction of nonstationarities could allow to better describe phenomena where local variations are clearly present. The approaches by \citet{fulgstad2015} and \citet{pereira_desassis_allard_2022} have being investigated and generalized to the spatio-temporal framework, but are left for a follow-up publication. In a nonstationary context, a comparison to the echo state networks propossed in \citet{huang22} would also be of great interest.

Another interesting consequence of defining the models through local stochastic partial differential equations is that the SPDEs still make sense when $\bbR^d$ is replaced by a space that is only locally flat. We can define nonstationary Gaussian fields on manifolds, and still obtain a GMRF representation. Important improvements were obtained in the spatial case \citep{pereira_desassis_allard_2022}. The generalization to space-time processes is being explored further \citep{pereira2023hal}.

Possible generalization to spatio-temporal SPDEs with a fractional exponent in the diffusion term could also be considered. A development of the methods proposed by \citet{bolin2019} and \citet{vabishchevich2015} should be explored.

\backmatter

\bmhead{Acknowledgments}

We are grateful to the O.I.E. center of Mines Paris, especially to Yves-Marie Saint-Drenan, Philippe Blanc and Hadrien Verbois, for providing the data and for inspiring discussions about renewable resources evaluation. We thank the Mines Paris / INRAE chair ``Geolearning'' for its constant support. We are thankful to the STSDS department of King Abdullah University of Science and Technology (KAUST), and especially to Professor Marc Genton, for the insightful work carried out during a stay in Saudi Arabia.

\clearpage
\begin{appendices}


\section{Relation between advection-diffusion SPDE and infinite dimensional SDE}
\label{sec:daprato}

Let us consider the advection-diffusion SPDE \eqref{eq:adv_diff}, here recalled:

\begin{gather}
    \left[\frac{\partial}{\partial t} + \frac{1}{c}(\kappa^2  - \nabla \cdot \Hbold\nabla)^{\alpha} + \frac{1}{c}\gammabold \cdot \nabla \right] X(t,\sbold) = \frac{\tau}{\sqrt{c}} Z(t,\sbold), \label{eq:adv_diff_app}    \\
    Z(t,\mathbf{s}) = W_T(t) \otimes  Z_S(\mathbf{s}).\label{eq:noise_app}
\end{gather}

The spatio-temporal forcing as defined in Equation \eqref{eq:noise_app} is a GeRF on functions of $L^2([0,T]) \times L^2(\bbR^d)$. 
Then, for any $(\phi_T, \phi_S),(\psi_T, \psi_S) \in L^2([0,T]) \times L^2(\bbR^d)$, $(W_T\otimes Z_S)(\phi_T, \phi_S)$ and $(W_T\otimes Z_S)(\psi_T, \psi_S)$ are centered Gaussian random variables with 

\begin{equation*}
\cov[(W_T\otimes Z_S)(\phi_T, \phi_S),(W_T\otimes Z_S)(\psi_T, \psi_S)]=\int_0^T \phi_T \psi_T \dd{t} \int_{\bbR^d} Z_S(\phi_S) Z_S(\psi_S) \dd{\sbold}.
\end{equation*}

The space-time forcing term $ W_T \otimes  Z_S$ can be identified with a cylindrical Wiener process $\{\widetilde{W}_t\}_{t\in[0,T]}$ in $L^2(\bbR^d)$ through
$$\widetilde{W}_t(\phi_S) = (W_T\otimes Z_S)(\mathbb{1}_{[0,t]}, \phi_S), \qquad \phi_S\in L^2(\bbR^d), \quad t \in [0,T],$$
where $\mathbb{1}_{[0,t]}$ is the indicator function over $[0,t]$.
Moreover, we can write
$$(W_T\otimes Z_S)(\phi_T, \phi_S) = \int_{\bbR^d} \left(\int_0^T \phi_T\dd{\widetilde{W}_t} \right)\phi_S \dd{\sbold} \qquad \phi_S\in L^2(\bbR^d), \quad t \in [0,T],$$
which leads to the interpretation of $W_T\otimes Z_S$ as the time derivative of the cylindrical Wiener process $\{\widetilde{W}_t\}_{t\in[0,T]}$.
This analogy allows us to rewrite the SPDE \eqref{eq:adv_diff_app} in the SDE form in infinite dimensions \citep{daprato1992}  

$$dX = - \frac{1}{c}\left[(\kappa^2  - \nabla \cdot \Hbold\nabla)^{\alpha} + \gammabold \cdot \nabla \right] X\dd{t} + \frac{\tau}{\sqrt{c}} \dd{\widetilde{W}_t}.$$

\section{Proof of Proposition \ref{prop:spatial_trace}}\label{sec:A0}

We present here the proof of Proposition \ref{prop:spatial_trace}.

\begin{proof}

The covariance function of the spatial trace between $X(t,\sbold)$ and $X(t,\sbold')$ for a spatial lag $\hbold=\sbold-\sbold'$ does not depend on the imaginary part of the spatial symbol function \citep{vergara2022general}, hence it can be written as
\begin{align}
    \cov(0,\hbold) & = \int_{\bbR^d} \int_{\bbR} \exp(i\hbold\xibold)S(\xibold,\omega) \dd{\omega} \dd{\xibold} \nonumber \\
    & = \int_{\bbR^d} \exp(i\hbold\xibold) \left[ \int_{\bbR} S(\xibold,\omega) \dd{\omega}\right] \dd{\xibold} \nonumber\\
    & = \int_{\bbR^d} \exp(i\hbold\xibold) S_S(\xibold) \dd{\xibold} ,
    \label{eq:covariance}
\end{align}
where $S(\xibold,\omega)$ is the spectral density defined as
$$ S(\xibold,\omega) = \frac{\tau^2}{(2\pi)^{(d+1)}\left[\omega^2 + c^{-2}(\kappa^2 + \xibold^\top\Hbold\xibold)^{2\alpha}\right]c(\kappa^2 + \xibold^\top\Hbold\xibold)^{\alpha_S}}.$$
Integrating over $\omega$, we obtain the spatial spectral density 
\begin{align}
S_S(\xibold) & = \frac{\tau^2}{(2\pi)^d c(\kappa^2+\xibold^\top\Hbold\xibold)^{\alpha_S}} \int_{\bbR} \frac{1}{2\pi\left[\omega^2 + c^{-2} (\kappa^2+\xibold^\top\Hbold\xibold)^{2\alpha}\right]} \dd{\omega} \nonumber \\
 & = \frac{\tau^2}{(2\pi)^d c(\kappa^2+\xibold^\top\Hbold\xibold)^{\alpha_S}} \frac{1}{2\left[c^{-2}(\kappa^2+\xibold\Hbold\xibold)^{2\alpha}\right]^{1/2}}\nonumber \\
 & = \frac{\tau^2}{2(2\pi)^d (\kappa^2+\xibold^\top\Hbold\xibold)^{\alpha_{\text{tot}}}}.
 \label{eq:spectrum}
\end{align}
Using the change of variable $\xibold=\kappa\Hbold^{-1/2}\wbold$ and plugging Equation \eqref{eq:spectrum} into \eqref{eq:covariance}, we obtain
\begin{align*}
\cov(0,\hbold) & = \frac{\tau^2}{2} \int_{\bbR^d} \frac{e^{i \hbold\xibold}}{(2\pi)^d (\kappa^2+\xibold^\top\Hbold\xibold)^{\alpha_{\text{tot}}}} \dd{\xibold}\\
& = \frac{\tau^2}{2} \int_{\bbR^d} \frac{e^{i \hbold\kappa\Hbold^{-1/2}\wbold}\lvert\kappa\Hbold^{-1/2}\rvert}{(2\pi)^d (\kappa^2+\kappa^2\wbold^\top\wbold)^{\alpha_{\text{tot}}}} \dd{\wbold}\\
& = \frac{\tau^2}{2\kappa^{2(\alpha_{\text{tot}}-d/2)}\lvert\Hbold
\rvert^{1/2}} \int_{\bbR^d} \frac{e^{i \hbold\kappa\Hbold^{-1/2}\wbold}}{ (2\pi)^d(1+\wbold^\top\wbold)^{\alpha_{\text{tot}}}} \dd{\wbold}\\
& = \frac{\tau^2\Gamma(\alpha_{\text{tot}}-d/2)}{2\Gamma(\alpha_{\text{tot}})(4\pi)^{d/2}\kappa^{2(\alpha_{\text{tot}}-d/2)}\lvert\Hbold
\rvert^{1/2}}
C^M_{\alpha_{\text{tot}}-d/2}\left(\kappa \norm{\Hbold^{-1/2}\hbold}\right).
\end{align*}
The last result comes from the computation of $\int_{\bbR^d} (1+\wbold^\top\wbold)^{-\alpha_{\text{tot}}}\dd{\wbold}$ with polar coordinates.
\end{proof}

\section{Discretization of spatio-temporal advection-diffusion SPDE}\label{sec:A1}

We detail here the discretization scheme of the advection-diffusion spatio-temporal SPDE  \eqref{eq:adv_diff}. 

For the sake of a clearer exposition, we set $\Hbold=\Ibold$, $\alpha=1$ and we consider a spatio-temporal white noise $Z(t,\mathbf{s})=W(t,\mathbf{s})$. The proof for the general case follows exactly the same lines as the proof below. The considered SPDE is
\begin{equation}
    \left[\frac{\partial}{\partial t} + \frac{1}{c}(\kappa^2 - \Delta) + \frac{1}{c}\gammabold \cdot \nabla \right] X(t,\sbold) = \frac{\tau}{\sqrt{c}} W(t,\sbold).
    \label{eq:diff_adv_spde}
\end{equation}

For the discretization of the temporal derivative in Equation \eqref{eq:diff_adv_spde}, we opt for the implicit Euler scheme, which considers the differential equation
$$\frac {\partial X(t)}{\partial t}=f(t,X),$$
with initial value $X^{(0)}=X(t_0)$. The method produces a sequence $\{X^{(k)}\}_{k=0}^{N_T}$, such that $X^{(k)}$ approximates $X(t_{0}+k dt)$, where $dt$ is the time step size. The approximation reads
\begin{equation}
    X^{(k+1)}=X^{(k)}+ dt f(t^{(k+1)},X^{(k+1)}).
\label{eq:implic_euler}
\end{equation}
In the specific case of Equation \eqref{eq:diff_adv_spde}, the implicit Euler discretization step reads

\begin{equation}
X^{(k+1)}(\sbold) -X^{(k)}(\sbold) + dt\left[\frac{1}{c} (\kappa^2-\Delta) + \frac{1}{c}\gammabold \cdot \nabla \right]X^{(k+1)}(\sbold)= \frac{\sqrt{dt}\tau}{\sqrt{c}} W^{(k+1)}_{S}(\sbold),
\label{eq:spde_discr_temp}
\end{equation} 
where $W^{(k+1)}_{S}(\sbold)$ is a spatial white noise obtained by integrating out the temporal white noise.

For ease of notation, we denote $X^{(k+1)}=X^{(k+1)}(\sbold)$, $X^{(k)}=X^{(k)}(\sbold)$ and $W_{S}=W^{(k+1)}_{S}(\sbold)$, since the spatial noise is independent of the temporal step $k$.

At each time step of the temporal discretization, a spatial FEM method is applied. In our case, we use the continuous Galerkin with Neumann boundary condition.  The weak form of Equation \eqref{eq:spde_discr_temp} is
\begin{equation}
\begin{split}
\int_\Omega  X^{(k+1)} v\dd{\sbold} &+ \frac{dt}{c}\left(\int_\Omega\kappa^2 X^{(k+1)} v \dd{\sbold} - \int_\Omega  \Delta X^{(k+1)} v \dd{\sbold}
+ \int_\Omega \gammabold  \cdot \nabla X^{(k+1)} v \dd{\sbold} \right) \\
&= \int_\Omega X^{(k)}v \dd{\sbold} + \frac{\sqrt{dt}\tau}{\sqrt{c}} \int_\Omega vW_{S}(\dd{\sbold}), \quad \forall v\in \mathcal{V},
\end{split}
\label{eq:weak_form}
\end{equation}
where $\mathcal{V}$ is the Hilbert space in which we search the solution and $W_{S}(v)$ is the white noise applied to the test function $v$.

By applying Green's first identity, i.e., by writing
$$\int_\Omega  \Delta X^{(k+1)}v \dd{\sbold} = - \int_\Omega  \nabla X^{(k+1)} \cdot \nabla v \dd{\sbold} +  \int_{\partial\Omega} v \cdot(\nabla X^{(k+1)} \cdot \hat{\mathbf{n}}) \dd{\sigma},$$
with $\hat{\mathbf{n}}$ being the normal vector on the boundary, and by simplifying the second term thanks to the Neumann boundary condition, we obtain
\begin{equation*}
\begin{split}
&\underbrset{ \mathcal{A}(X^{(k+1)},v) }{\int_\Omega  X^{(k+1)} v \dd{\sbold} + \frac{dt}{c}\left(\int_\Omega  \kappa^2 X^{(k+1)} v \dd{\sbold} + \int_\Omega  \nabla X^{(k+1)} \cdot \nabla v \dd{\sbold} + \int_\Omega \gammabold  \cdot \nabla X^{(k+1)} v \dd{\sbold} \right)}\\
&\phantom{Xxxxxxxxxxxxx}=\underbrset{ \mathcal{C}(X^{(k)},v) }{\int_\Omega X^{(k)}v \dd{\sbold}} + \underbrset{ \mathcal{E}(v) }{\frac{\sqrt{dt}\tau}{\sqrt{c}} \int_\Omega vW_{S}(\dd{\sbold})}, \quad \forall v\in \mathcal{V}.
\end{split}
\end{equation*}

Let $\mathcal{V}_h$ be the space of finite element solutions spanned by the basis functions $\{\psi_i\}_{i=1}^{N_S}$.
The Galerkin method allows us to find an approximated solution $X^{(k+1)}_{h} \in \mathcal{V}_h \subset \mathcal{V}$ to the SPDE, such that
\begin{equation}
    \mathcal{A}(X^{(k+1)}_{h},v_h) =\mathcal{C}(X^{(k)}_{h},v_h) + \mathcal{E}(v_h) \quad \forall v_h\in \mathcal{V}_h.
    \label{eq:gen_galerkin}
\end{equation}

The functions $X^{(k+1)}_{h}$, $X^{(k)}_{h}$ and $v_h$ are linear combinations of the basis functions, with
$$X^{(k+1)}_{h} = \sum_{i=1}^{N_S} x^{(k+1)}_{i} \psi_i; \quad X^{(k)}_{h} = \sum_{i=1}^{N_S} x^{(k)}_{i} \psi_i; \quad v_h= \sum_{i=1}^{N_S} v_{i} \psi_i.$$
Because of the linearity in the first argument of $\mathcal{A}(\cdot,\cdot)$ and $\mathcal{C}(\cdot,\cdot)$, we get
\begin{equation}
\sum_{i=1}^{N_S} \mathcal{A}(\psi_i, v_h)x^{(k+1)}_{i} = \sum_{i=1}^{N_S} \mathcal{C}(\psi_i, v_h)x^{(k)}_{i} + \mathcal{E}(v_h), \quad \forall v_h\in \mathcal{V}_h,
\label{eq:fem_form}
\end{equation}
where
\begin{align*}
&\mathcal{A}(\psi_i, v_h) = \mathcal{M}(\psi_i, v_h) + \frac{dt}{c}\left(\mathcal{K}(\psi_i,v_h) + \mathcal{B}(\psi_i,v_h)\right)\\
&\mathcal{C}(\psi_i, v_h)= \mathcal{M}(\psi_i, v_h),
\end{align*}
with $\mathcal{K}(\psi_i, v_h) = \kappa^2 \mathcal{M}(\psi_i, v_h) + \mathcal{G}(\psi_i, v_h)$. Here, $\mathcal{M}$ and $\mathcal{G}$ are the mass and stiffness operators, respectively $\mathcal{M}(v,w)=\int_\Omega vw \dd{\sbold}$ and $\mathcal{G}(v,w)=\int_\Omega  \nabla v\cdot \nabla w \dd{\sbold}$. $\mathcal{B}$ is the advection operator, i.e., $\mathcal{B}(v,w)=\int_\Omega \gammabold \cdot \nabla v w \dd{\sbold}$. Finally, $\mathcal{E}$ is the operator of the form $\mathcal{E}(v)=\frac{\sqrt{dt}\tau}{\sqrt{c}}\int_\Omega vW_{S}(\dd{\sbold})$.

Since any $v_h$ can be written as a linear combination of basis functions, the formulation \eqref{eq:fem_form} is equivalent to
\begin{equation}
\sum_{i=1}^{N_S} \mathcal{A}(\psi_i, \psi_j)x^{(k+1)}_{i} = \sum_{i=1}^{N_S} \mathcal{C}(\psi_i, \psi_j)x^{(k)}_{i} + \mathcal{E}(\psi_j), \quad \forall j.
\label{eq:fem_form_j}
\end{equation}

We define $\Mbold=[M_{ij}]_{i,j=1}^{N_S} = [\mathcal{M}(\psi_i,\psi_j)]_{i,j=1}^{N_S}$, $\Gbold=[G_{ij}]_{i,j=1}^{N_S} = [\mathcal{G}(\psi_i,\psi_j)]_{i,j=1}^{N_S}$, $\Bbold=[B_{ij}]_{i,j=1}^{N_S} = [\mathcal{B}(\psi_i, \psi_j)]_{i,j=1}^{N_S}$ the mass, stiffness and advection matrices, respectively.

$\mathcal{E}(\psi_j)$ is a Gaussian random variable with expectation 0 and covariance equal to 
\begin{eqnarray*}
\cov(\mathcal{E}(\psi_i), \mathcal{E}(\psi_j)) & = & \frac{dt\tau^2}{c}\cov\left[W_S(\psi_i),W_S(\psi_j)\right]\\
& = & \frac{dt\tau^2}{c}\int_\Omega \psi_i \psi_j \dd{\sbold} = \frac{dt\tau^2}{c}M_{ij},
\end{eqnarray*}
by following the definition of white noise in Equation \eqref{eq:white_noise}.

If $\mathbf{z}^{(k+1)}$ is a $(N_S)$-Gaussian vector such that $\mathbf{z}^{(k+1)} \sim \mathcal{N}(\0bold,\Ibold_{N_S})$, $\xbold^{(k+1)}$ is the vector containing the values $\{x^{(k+1)}_{i}\}_{i=1}^{N_S}$ and $\xbold^{(k)}$ is the vector containing the values $\{x^{(k)}_{i}\}_{i=1}^{N_S}$, then the sparse linear system corresponding to Equation \eqref{eq:fem_form_j} reads
\begin{equation}
\Mbold\mathbf{x}^{(k+1)} + \frac{dt}{c}(\Kbold+\Bbold) \mathbf{x}^{(k+1)}= \Mbold\mathbf{x}^{(k)} + \frac{\sqrt{dt}\tau}{\sqrt{c}}\Mbold^{1/2} \mathbf{z}^{(k+1)},
\label{eq:schema_spde_adv_reac_diff}
\end{equation} 
where $\Kbold=\kappa^2 \Mbold + \Gbold$ and $\Mbold^{1/2}$ is any matrix such that $\Mbold^{1/2}\Mbold^{1/2} = \Mbold$. 

When the spatial noise is colored, i.e. $Z_S(\sbold)$, the right-hand side operator $\mathcal{E}(v)$ becomes $\mathcal{E_S}(v)$, defined as
$$\mathcal{E_S}(v)=\frac{\sqrt{dt}\tau}{\sqrt{c}}\int_\Omega vZ_S(\dd{\sbold})$$
and it satisfies
$$\mathcal{E_S}(v_h) = \sum_{i=1}^{N_S}\mathcal{M}(\psi_i,v_h)z_{S,i}.$$
Hence,
$$\sum_{i=1}^{N_S} \mathcal{A}(\psi_i, \psi_j)x^{(k+1)}_{i} = \sum_{i=1}^{N_S} \mathcal{C}(\psi_i, \psi_j)x^{(k)}_{i} + \sum_{i=1}^{N_S}\mathcal{M}(\psi_i,\psi_j)z_{S,i}, \quad \forall j.$$
If $\zbold_S = \{z_{S,i}\}_{i=1}^{N_S}$ has precision matrix equal to $\Qbold_S$, then the sparse linear system is
\begin{equation}
\Mbold\mathbf{x}^{(k+1)} + \frac{dt}{c}(\Kbold+\Bbold) \mathbf{x}^{(k+1)}= \Mbold\mathbf{x}^{(k)} + \frac{\sqrt{dt}\tau}{\sqrt{c}}\Mbold \Lbold_S^\top \mathbf{z}^{(k+1)},
\label{eq:schema_spde_adv_reac_diff_spat}
\end{equation} 
where $\mathbf{z}^{(k+1)}\sim\mathcal{N}(\0bold,\Ibold_{N_S})$ and $\Lbold_S$ is the Cholesky decomposition of $\Qbold_{S}^{-1}$.

\begin{remark}
When the diffusion term includes an anisotropy matrix $\Hbold$, i.e., when $\Delta$ is replaced by $\nabla \cdot \Hbold\nabla$, the stiffness operator becomes $\mathcal{G}(v,w)=\int_\Omega  \Hbold \nabla v\cdot \nabla w \dd{\sbold}$, and the stiffness matrix changes consequently.
\end{remark}


\section{Proof of Proposition \ref{prop:global_prec}}\label{sec:A3}
We present here the proof of Proposition \ref{prop:global_prec}.

\begin{proof}
Let us denote $\xbold_{0:N_T}=[\mathbf{x}^{(0)},\dots,\mathbf{x}^{(N_T)}]^\top$ the vector containing all spatial solutions until time step $N_T$. Then, $$\xbold_{0:N_T}=\Rbold \begin{pmatrix}\mathbf{x}^{(0)}\\ \zbold_{1:N_T}\end{pmatrix},$$
with $\zbold_{1:N_T}=[\mathbf{z}^{(1)},\dots,\mathbf{z}^{(N_T)}]^\top$ and 
$$\Rbold = \begin{pmatrix}
\Ibold_{N_S}   & 0   & 0 & 0  &\dots& 0 \\
\Dbold & \Ebold & 0 &0  & \dots &0 \\
\Dbold^2 & \Dbold \Ebold & \Ebold & 0& \dots& 0\\
\vdots &\ddots &\ddots &\ddots & \ddots &  \vdots \\
\vdots &\ddots &\ddots &\ddots & \ddots &  0 \\
\vdots &\ddots &\ddots & \Dbold^2 & \Dbold &  \Ebold \\
\end{pmatrix}.$$

$\Rbold$ has a block structure which allows easy computation of its inverse
$$\Rbold^{-1} = \begin{pmatrix}
\Ibold_{N_S}   & 0   & 0 & 0  &\dots& 0 \\
-\Ebold^{-1}\Dbold & \Ebold^{-1} & 0  & 0 & \dots &0 \\
0 & -\Ebold^{-1}\Dbold & \Ebold^{-1} & 0 &\dots& 0\\
\vdots &\ddots &\ddots &\ddots & \ddots & \vdots \\
\vdots &\ddots &\ddots &\ddots & \ddots & 0 \\
0 &\dots &\dots & 0 & -\Ebold^{-1}\Dbold & \Ebold^{-1} \\
\end{pmatrix}.$$

The precision matrix of $\xbold_{0:N_T}$ is thus
$$\Qbold={\Rbold^{-1}}^\top\begin{pmatrix}
\Sigmabold^{-1} & 0 & \dots & 0\\
0  & \Ibold_{N_S} &\dots & 0\\
\vdots & \ddots & \ddots & \vdots\\
0 & 0 & \dots & \Ibold_{N_S}
\end{pmatrix} \Rbold^{-1}.$$

By denoting $\Fbold=\Ebold \Ebold^\top$, the global precision matrix reads

$$
\scalemath{0.85}{
\Qbold = \begin{pmatrix}
\Sigmabold^{-1}+\Dbold^\top\Fbold^{-1}\Dbold   & -\Dbold^\top\Fbold^{-1}   & 0 & \dots & 0 \\
-\Fbold^{-1}\Dbold & \phantom{X} \Fbold^{-1}+\Dbold^\top\Fbold^{-1}\Dbold \phantom{X} & -\Dbold^\top\Fbold^{-1}  & \ddots & \vdots \\
\vdots & \ddots & \ddots & \ddots & 0 \\
\vdots & \ddots & -\Fbold^{-1}\Dbold & \phantom{X}\Fbold^{-1}+\Dbold^\top\Fbold^{-1}\Dbold \phantom{X} & -\Dbold^\top\Fbold^{-1} \\
0 & \dots  & 0 & -\Fbold^{-1}\Dbold & \Fbold^{-1}
\end{pmatrix}.}
$$

By replacing the values of $\Dbold$ and $\Fbold$ and by defining $\Jbold = \left[\Mbold + \frac{dt}{c}(\Kbold +\Bbold + \Sbold)\right]$, we obtain
$$
\scalemath{0.65}{
\Qbold = \frac{c}{\tau^2 dt}\begin{pmatrix}
\Sigmabold^{-1}+\Qbold_S  & -\Qbold_S\Mbold^{-1}\Jbold   & 0 & \dots & 0 \\
-\Jbold^\top\Mbold^{-1}\Qbold_S & \phantom{X} \Jbold^\top \Mbold^{-1}\Qbold_S\Mbold^{-1}\Jbold + \Qbold_S \phantom{X} & -\Qbold_S\Mbold^{-1}\Jbold  & \ddots & \vdots \\
\vdots & \ddots & \ddots & \ddots & 0 \\
\vdots & \ddots & -\Jbold^\top\Mbold^{-1}\Qbold_S &  \phantom{X} \Jbold^\top \Mbold^{-1}\Qbold_S\Mbold^{-1}\Jbold + \Qbold_S \phantom{X} & -\Qbold_S\Mbold^{-1}\Jbold \\
0 & \dots  & 0 & -\Jbold^\top\Mbold^{-1}\Qbold_S & \Jbold^\top \Mbold^{-1}\Qbold_S\Mbold^{-1}\Jbold
\end{pmatrix}.}
$$  
\end{proof}

\end{appendices}

\clearpage
\bibliographystyle{apalike}
\bibliography{sn-bibliography_new}

\end{document}